\documentclass[sigconf, nonacm]{acmart}

\usepackage{pvldb}

\renewcommand\vldbdoi{XX.XX/XXX.XX}
\renewcommand\vldbpages{XXX-XXX}
\renewcommand\vldbavailabilityurl{https://github.com/dtak-upc/DIADA} 

\usepackage{makecell}
\usepackage{multirow}
\usepackage{algorithm}
\usepackage[noend]{algpseudocode}

\usepackage{todonotes}
\usepackage{enumitem}
\usepackage{pifont}

\newcommand{\sys}{DIADA}
\makeatletter
\newcommand{\stitle}[1]{\noindent\textbf{#1\@addpunct{.}}}
\newcommand{\ititle}[1]{\noindent\emph{#1\@addpunct{.}}}
\makeatother
\newcommand{\cmark}{\ding{51}}
\newcommand{\pmark}{$\circ$}
\newcommand{\xmark}{\ding{55}}

\theoremstyle{definition} 

\newtheorem{problem}{Problem}
\newtheorem{assumption}{Assumption}
\newtheorem{proposition}{Proposition}

\begin{document}
\title{DIADA: Automatic Data Composition in Data Lakes}

\author{Marc Maynou}
\orcid{-}
\affiliation{%
  \institution{Universitat Politècnica de Catalunya}
  \city{Barcelona}
  \country{Spain}
}
\email{marc.maynou@upc.edu}

\author{Albert Martin}
\orcid{-}
\affiliation{%
  \institution{Universitat Politècnica de Catalunya}
  \city{Barcelona}
  \country{Spain}
}
\email{albert.martin.g@upc.edu}

\author{Sergi Nadal}
\orcid{-}
\affiliation{%
  \institution{Universitat Politècnica de Catalunya}
  \city{Barcelona}
  \country{Spain}
}
\email{sergi.nadal@upc.edu}

\author{Anna Queralt}
\orcid{-}
\affiliation{%
  \institution{Universitat Politècnica de Catalunya}
  \city{Barcelona}
  \country{Spain}
}
\email{anna.queralt@upc.edu}

\author{Oscar Romero}
\orcid{-}
\affiliation{%
  \institution{Universitat Politècnica de Catalunya}
  \city{Barcelona}
  \country{Spain}
}
\email{oscar.romero@upc.edu}

\begin{abstract}
Data lakes contain a plethora of attributes scattered across many tables that, when combined, provide enhanced assets for data analysis. Nonetheless, deciding which attributes belong together in meaningful relations remains a manual, per-task effort. Merging by joinability alone provides no guarantees regarding attribute relevance, while selecting features against a single target discards attributes useful to other tasks. To address this gap, we introduce the data composition problem: organizing a fragmented, heterogeneous lake into meaningful relations, agnostic of any particular analytical task so that the resulting organization can serve as a common foundation for diverse downstream analyses. We propose \sys, a composition system that employs multivariate dependence as the criterion for assessing the meaningfulness of a relation and approximates it by hypothesizing independence among attributes and identifying those sets that violate this hypothesis. To do so, we map the attributes to a predicate space, forming a lattice under inclusion and mining those predicate sets that exhibit dependence among their constituents. We contribute a dedicated and scalable algorithm to effectively explore this space, outscaling classical algorithms for mining relationships, thus discovering dependencies that would otherwise be impractical to identify. We demonstrate that applying a single data composition process benefits diverse potential downstream tasks. This is the result of providing a subset of low-noise, statistically relevant attributes that increases the confidence that detected patterns are grounded in real relationships, thus preventing common modeling issues in large-scale environments.

\end{abstract}

\maketitle

\vldbtopmatter

\section{Introduction} \label{sec:introduction}
Modern data-driven organizations store their data in massive, heterogeneous repositories known as data lakes \cite{hai2023data}. Their flexible ingestion policies, however, often turn these repositories into poorly organized \textit{data swamps} where finding and preparing relevant data becomes cumbersome \cite{nargesian2019data}. To mitigate this problem, organizations structure their data management pipelines into successive curation zones (e.g., the notions of bronze, silver, and gold layers used in the industry \cite{databricks2023medallion}) where data assets undergo increasing levels of refinement before being made available for downstream analytics \cite{armbrust2021lakehouse}. 
However, automatically organizing the resulting assets to favor their subsequent use is a challenging task. We focus on this particular stage of the process: once data from different sources have been homogenized, how their attributes should be integrated into relevant relations for downstream analysis. We call this problem \textbf{data composition}: organizing the attributes of a fragmented data lake into \textit{meaningful} relations.

Consider Figure \ref{fig:discovery_example}a, which depicts a simplified data lake whose attributes are color coded according to their relevance to a specific target variable. For example, both the \textit{num\_rooms} of a property and the \textit{tax\_rate} of its neighborhood are likely to be informative to predict its \textit{market\_value}, whereas \textit{agent\_phone\_num} is not related to either \textit{market\_value} or any other attribute. A naive approach to structure the data would run a data discovery task to identify candidate join attributes and perform a multiway join (Figure \ref{fig:discovery_example}b) that combines potentially useful information scattered across datasets into a single relation. However, data discovery may propose spurious joins, such as \textit{property\_id} and \textit{station\_id}, alongside meaningful ones, such as the join on \textit{postal\_code}. Consequently, the relation supplied to an ML pipeline predicting, say, \textit{market\_value}, would contain inert attributes (gray) and an entire group of weather attributes (blue) irrelevant to this particular prediction task. 
Ideally, the joined table should be reorganized as shown in Figure \ref{fig:discovery_example}c, with intrinsically related attributes composed into separate relations. 

\begin{figure*}[h]
	\begin{center}
		\includegraphics[width=1\linewidth]{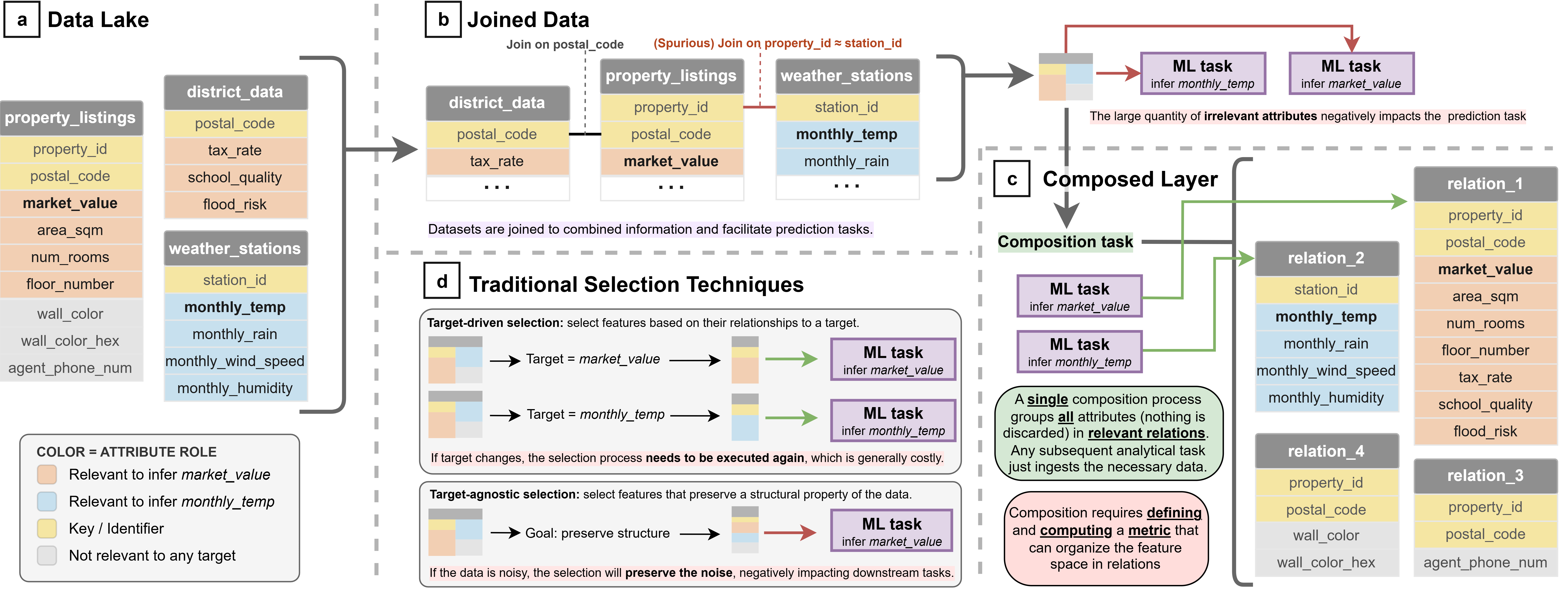}
        \Description{Diagram in three panels showing a fragmented data lake, the join candidates it produces including a spurious join, and the composed relations grouping attributes by multivariate dependence.}
        \caption{The data composition problem illustrated. A data lake (a) contains tables that can be used to satisfy predictive tasks. A multiway join is executed for these tables (b) to place all relevant information in a single asset, but a spurious join introduces unhelpful features in the dataset. Ideally, joined data should be separated into specific relations (c) to benefit any potential downstream task. In practice, this process is either done manually (costly and prone to errors) or substituted by automated techniques (d), which either select features for a specific target (and discard the rest) or preserve noisy data.}
		\label{fig:discovery_example}
	\end{center}
\end{figure*}

The task of data composition is mainly performed in a manual, slow and unreliable fashion by domain experts through the execution of handcrafted joins and partitions \cite{abadi2020seattle}. Alternatively, there are many research fields that propose automated methodologies to organize the attributes and relations of a data repository, each presenting a different way of characterizing the \textit{meaningfulness} of the final relations. Nonetheless, all have important limitations towards performing data composition tasks. Functional dependency (FD) discovery identifies attribute sets that determine others \cite{huhtala1999tane}, which can guide the separation of datasets into normalized relations. However, these constraints are deterministic, and the classical approaches to detect them are syntactic \cite{DBLP:journals/pvldb/PapenbrockEMNRZ15} and scale poorly \cite{DBLP:conf/sigmod/IlyasMHBA04}, making them not suitable for data lakes' heterogeneity and large size. Data discovery \cite{fernandez2018seeping, nargesian2019data} presents semantic and scalable techniques to identify tables that are related and can be combined, but does not assess the relevance of the attributes in the relation. This leads to noise not being filtered out, as the spurious join in Figure \ref{fig:discovery_example}b illustrates. To address this, we can introduce techniques to select the relevant subset of attributes out of the combined data. Figure \ref{fig:discovery_example}d exemplifies two main families of such approaches. On one hand, the selection of attributes can be guided by the association to a predictive target (i.e., target-driven, such as data augmentation and supervised selection \cite{DBLP:journals/pvldb/ChepurkoMZFKK20, DBLP:journals/tmlr/CappuzzoCLPV25, DBLP:conf/icde/IonescuVB0K24}). However, assessing attribute relevance against a given target discards whatever is not associated with it, yet these deleted attributes might provide information that a different task may need (e.g., blue features for predicting \textit{monthly\_temp}). Thus, the costly selection process must be executed anew for every new target. On the other hand, target-agnostic approaches (e.g., unsupervised feature selection \cite{mitra2002unsupervised, solorio2020review}) circumvent this problem by not conditioning attribute importance on a target, instead seeking to preserve some structural property of the data (e.g., cluster structure). However, as they do not explicitly assess the degree of shared information across attributes, uninformative attributes may be preserved if they contribute to the original and potentially noisy structure being captured.
At the core of this problem lies the need for assessing the \textit{relatedness} of sets of attributes and, therefore, guide the separation of features into strongly intra-related groups. We argue that relatedness is, fundamentally, a question of shared information: two attributes are related when knowledge of one changes what is expected of the other. This notion extends naturally from pairs of attributes to arbitrary sets through their multivariate dependence \cite{watanabe1960information}, that is, the extent to which each attribute is conditioned by all the others jointly. This dependence is precisely what analytical tasks look for, as it allows patterns to be defined and exploited to extract insight or perform inference. A meaningful relation is, therefore, one whose attributes exhibit strong multivariate dependence. However, metrics that fully compute such a notion (i.e., total correlation \cite{watanabe1960information} and its variants \cite{DBLP:journals/tit/McGill54}) are exponential to estimate due to the combinatorial explosion of assessing increasingly large sets of features jointly \cite{bell2003co}. Thus, current methods to compute multivariate dependence either fail to scale or opt exclusively for pairwise dependencies. Take a property's \textit{floor\_number}, its district's \textit{flood\_risk}, and its \textit{market\_value} from Figure \ref{fig:discovery_example}: the marginal effects of both \textit{floor\_number} and \textit{flood\_risk} over \textit{market\_value} are likely to be minimal and, hence, missed. Nonetheless, ground floors are cheaper where flooding is a risk. Thus, only when both are observed does the strong dependence with \textit{market\_value} appear. 

In this work we aim at providing an approach that efficiently approximates multivariate dependencies across sets of attributes at scale. This methodology is encapsulated in \sys, the first system that performs automatic data composition based on evaluating the amount of shared information across sets of attributes. To do so, \sys~hypothesizes that all attributes in the relation are independent, and its goal is to detect which groups of attributes violate that independence hypothesis. \sys~uses an intermediate abstraction layer that decouples raw data formats from multivariate dependency analysis, allowing dependencies to be discovered across heterogeneous data. More precisely, every attribute is mapped to a set of predicates via a Boolean function over a chosen domain (e.g., individual tuples or pairs of tuples). The combination of these predicates results in a space that forms a lattice under inclusion. \sys~then mines this lattice to extract predicate sets whose joint probability departs from what is expected by assuming predicate independence. Such a departure indicates that the corresponding attributes are related and should be composed into the same relation. \sys, thus, approximates multivariate dependence by applying a statistical test regarding the independence of predicates.

Exploring such a lattice reduces to to the well-studied problem of mining frequent itemsets. Nonetheless, scaling this search to an entire data lake still introduces computational challenges both in the number of rows (i.e., horizontal, as there are more predicates per attribute) and columns (i.e., vertical, as there are more combinations of predicates). To address this, \sys~introduces an \textit{approximation factor}, a statistically validated ceiling on how uncertain any estimated probability may remain. Thus, \sys~performs an informed and traceable search that prunes the space by navigating \textit{only} towards predicate sets whose relatedness is still highly uncertain. Consequently, our approach provides a statistics-based definition of meaningful relations in this context, and contributes an efficient and scalable method to approximate multivariate dependencies that allow us, for the first time, to operationalize the problem of data composition.

Our experiments showcase how these optimizations allow \sys~to outscale traditional miners in the context of data lakes and discover dependencies that could otherwise be impractical to detect. We prove that in realistic scenarios \sys~correctly filters noisy attributes from relations while preserving relevant attributes, which other approaches are not able to do. We also demonstrate that not composing data prior to executing downstream tasks negatively impacts the final models, as scalable analytical pipelines are not able to effectively preserve relevant data. By executing \sys~a single time, diverse downstream models improve both their performance and generalization, regardless of their target.

\medskip

\stitle{Contributions} Our contributions are as follows:
\begin{itemize}[leftmargin=1em, itemsep=0pt]
    \item We define the \textbf{data composition problem}: how to organize a fragmented lake into meaningful relations, agnostic of any target attribute. We take multivariate dependence as the composition criterion. 
    \item We introduce an \textbf{approximate} and \textbf{scalable algorithm} that mines a predicate-set lattice to detect sets of dependent attributes, outperforming classical methods for frequent itemset mining as attributes grow, discovering dependencies that would be otherwise impractical.
    \item We present \textbf{\sys}: the first system for automatic data composition, which implements the above algorithm. We showcase that applying a single, target-agnostic composition task improves the utility and robustness of downstream modeling pipelines regardless of the target attribute, as composing generates meaningful datasets not conditioned on a given process.
    
    
\end{itemize}

\section{Related Work} \label{sec:related_work}

The task of data composition can be further specified through four requirements, which need to be simultaneously fulfilled. (R1)~\textit{Target-agnostic}: composition cannot be defined with respect to a target variable. (R2)~\textit{Multivariate}: composition must be guided by the relatedness of sets of attributes, not merely attribute pairs. (R3)~\textit{Complete}: every attribute must be placed in some relation. (R4)~\textit{Lake-scale}: executing the composition must remain feasible as the lake grows. In the remainder of this section, we review the main families of approaches that are relevant to data composition. Table~\ref{tab:related} summarizes the extent to which each family satisfies the stated requirements.

\begin{table}[]
\caption{Requirements on the composition criterion.
\cmark: satisfied, \pmark: partially, \xmark: not addressed.}
\label{tab:related}
\small
\begin{tabular}{lcccc}
\toprule
& R1 & R2 & R3 & R4 \\
\midrule
Join / union discovery            & \cmark & \xmark & \xmark & \cmark \\
Lake organization                 & \cmark & \xmark & \pmark & \pmark \\
Augmentation / supervised FS      & \xmark & \pmark & \xmark & \pmark \\
Unsupervised feature selection    & \cmark & \pmark & \pmark & \cmark \\
FD / key discovery                & \cmark & \xmark & \cmark & \xmark \\
Statistical dependency discovery  & \cmark & \xmark & \cmark & \pmark \\
Schema decomposition              & \cmark & \cmark & \cmark & \xmark \\
\midrule
\sys                              & \cmark & \cmark & \cmark & \cmark \\
\bottomrule
\end{tabular}
\end{table}

\subsection{Assembling data from a lake}

\stitle{Data discovery} Join and union discovery locate tables that can be combined with a given base table through value overlap, containment or embedding similarity \cite{fernandez2018seeping, nargesian2019data, nargesian2018table, DBLP:journals/tkde/MaynouNPFRQ26, bogatu2020dataset}. These methods answer whether two tables \emph{can} be combined, not whether the attributes they contribute form a \textit{meaningful} relation (R2, R3). Their goal is to identify related assets, not to evaluate the relevance of the attributes. Discovery is therefore the natural upstream stage for \sys, and we use it as such (Section~\ref{sec:pipeline}).

\medskip

\stitle{Lake organization} A smaller line organizes the lake itself rather than serving individual queries. Aurum builds an enterprise knowledge graph linking columns by syntactic and semantic similarity \cite{DBLP:conf/icde/FernandezAKYMS18}, and navigation-oriented approaches induce a hierarchy over tables so that users can browse a lake without knowing its contents \cite{DBLP:journals/tkde/NargesianPBZM23}. These share our objective of a single reusable organization, but they organize tables by high-level similarity, rather than through the relatedness of their attributes (R2).

\medskip

\stitle{Data augmentation and supervised selection}
Supervised feature selection aims to identify a subset of attributes that maximizes predictive performance for a defined target variable, whereas data augmentation expands the search space to an entire data lake by integrating additional data sources before feature selection \cite{DBLP:journals/pvldb/ChepurkoMZFKK20, DBLP:conf/icde/IonescuVB0K24, DBLP:journals/tmlr/CappuzzoCLPV25, cui2026tabular, guyon2003introduction}. Because relevance is defined with respect to that target, the search is repeated whenever the analytical goal changes, and attributes useful to later tasks are discarded rather than organized (R1, R3).

\subsection{Unsupervised feature selection}

Unsupervised feature selection (UFS) reduces dataset dimensionality while preserving a surrogate objective over the data
\cite{mitra2002unsupervised, solorio2020review}. Methods differ in the structure they maintain: local geometry \cite{he2005laplacian, zhao2007spectral, li2012unsupervised}, cluster assignment \cite{cai2010unsupervised, DBLP:conf/ijcai/YangSMHZ11}, matrix reconstruction \cite{zhu2015unsupervised, balin2019concrete}, and global variance \cite{lim2021pairwise, yuan2021novel}. Their main limitation is that their structure-preserving objectives assume the structure is sound to begin with. Hence, if the input contains correlated noise, the objective preserves exactly what should be removed. Nonetheless, it is the only family of methods that fulfills, to some extent, all four composition requirements. Thus, in Section \ref{sec:lima_vs_ufs}, we quantify their suitability for the composition problem.

\subsection{Discovering dependencies}

\stitle{Exact dependencies} Functional dependency (FD) discovery identifies attribute sets that determine others, and underpin normalization, cleaning and query optimization \cite{huhtala1999tane, DBLP:journals/pvldb/PapenbrockEMNRZ15}. Modern FD discovery algorithms scale to tables of realistic size, but still have combinatorial worst-case behavior (R4). Moreover, FDs express deterministic constraints: an exact FD is invalidated by a single counterexample tuple, while the absence of an FD does not imply statistical independence. Consequently, FD and key discovery do not capture the more general multivariate relatedness between attributes (R2).

\medskip

\stitle{Statistical dependency discovery} The rigidity of exact constraints can be relaxed to discover broader statistical relationships between attributes. CORDS mines correlations and soft functional dependencies from samples \cite{DBLP:conf/sigmod/IlyasMHBA04}, and Zhang et al. give a statistical account of FD discovery under noise \cite{DBLP:conf/sigmod/ZhangGR20}. Both are target-agnostic and cheap enough for large tables, but they still lack the simultaneous assessment of the joint relatedness of sets of features (R2).


\medskip

\stitle{Schema decomposition} Maimon \cite{DBLP:conf/sigmod/KenigMPSS20} mines approximate acyclic schemes using information-theoretic measures, and is the only approach meeting the first three requirements. It solves, however, the dual problem. It refines a relation that already exists into a scheme that preserves it, so the set of attributes that belong together is an input and the question is how to split it. Additionally, Maimon is limited in terms of scalability (R4), its cost being exponential in the number of attributes and reported empirically up to 30 columns.

\subsection{Defining the composition criterion}\label{sec:rw-computing}

To effectively perform composition tasks it is necessary to establish and evaluate a statistical criterion for attribute relatedness. To that end, we position \sys~against two bodies of work.

\medskip

\stitle{Multivariate measures} Total correlation \cite{watanabe1960information}, interaction information, and its variants \cite{DBLP:journals/tit/McGill54} quantify the divergence between a joint distribution and the product of its marginals, and are the measures composition would ideally adopt. Estimating them is exponential in the number of variables \cite{bell2003co}, which is precisely what places any criterion defined through them outside (R4), and why scalable methods fall back on pairwise or greedy conditional formulations \cite{brown2012conditional, vergara2014review}. Section~\ref{sec:measures} examines these measures in the detail our framework requires and derives the alternative we adopt.

\medskip

\stitle{Pattern mining} Once dependence is expressed over predicates, the search space is the itemset lattice, for which Apriori \cite{DBLP:conf/sigmod/AgrawalIS93} and FP-Growth \cite{DBLP:conf/sigmod/HanPY00} are the canonical exploration algorithms. Sampling-based variants bound the deviation of the estimated supports \cite{DBLP:conf/vldb/Toivonen96, DBLP:conf/cikm/RiondatoDFU12}. We reuse the lattice but not the pruning principle, since we leverage independence rather than frequency. Certain approaches carry error control through significance testing with family-wise error correction \cite{DBLP:journals/ml/Webb07} and permutation-based methods \cite{DBLP:conf/kdd/PellegrinaV18}.
\section{A Statistical Criterion for meaning}
\label{sec:multivariate_dependencies}


This section develops the statistical criterion \sys~employs to quantify the \textit{meaningfulness} of a set of attributes. First, Section~\ref{sec:measures} fixes notation and the independence baseline. Then, Section~\ref{sec:predicate_spaces} introduces the predicate space that decouples heterogeneous data from multivariate dependency analysis, Section~\ref{sec:cohesion} identifies meaningfulness as a statistical hypothesis test with independence as its null, and Section~\ref{sec:problem_statement} formalizes the data composition problem.

\subsection{Preliminaries}\label{sec:measures}

Let $R(A_1, A_2, \dots, A_m)$ be a relational schema with attributes $A_1, A_2, \allowbreak \dots, A_m$, each taking values over domains $D_1, D_2, \dots, D_m$, respectively. A relation instance over this schema
is a finite set $r(R) \subseteq D_1 \times D_2 \times \cdots \times D_m$, where each tuple $t \in r(R)$ is a function mapping every attribute to a value in its corresponding domain, i.e., $t(A_i) \in D_i$.
A core assumption in statistics and machine learning is that the tuples $t = (a_1, a_2, \dots, a_m)$ of a relation are drawn from an unknown multivariate distribution over $D_1 \times D_2 \times \cdots \times D_m$~\cite{bishop2006pattern}.
When the attributes $A_1, A_2, \dots, A_m$ are independent, the joint probability of drawing any given tuple is reduced to the product of marginal attribute probabilities:
\[
  p(A_1 = a_1, A_2 = a_2, \dots, A_m = a_m) = \prod_{i} p(A_i = a_i).
\]
The goal of multivariate analysis is to discover sets of attributes that deviate from this independence. As argued in Section \ref{sec:introduction}, these deviations are what makes a relation worth preserving in data composition, as the shared information between attributes defines patterns that can be exploited by downstream tasks. Computing exact measures of multivariate dependence is unfeasible in the context of data lakes, so the criterion we require cannot be defined through them. It must instead make decisions under sampling while properly accounting for the associated uncertainty, be data-agnostic so as to operate on heterogeneous data, and provide objective measures of association that do not rely on arbitrary thresholds. The remainder of this section constructs this criterion.


\subsection{Predicate spaces}\label{sec:predicate_spaces}

Our framework generalizes statistical hypothesis testing to large, heterogeneous data lakes through an intermediate abstraction layer that decouples raw data from dependency analysis. It decomposes the problem of multivariate dependency discovery in two: 
\begin{itemize}[leftmargin=1em, itemsep=0pt]
    \item \textbf{Predicate space construction}: Given a relational schema of the form $R(A_1, A_2, \dots, A_m)$, generate an abstraction layer known as the predicate space.
    \item \textbf{Predicate association mining}: Given the predicate space and a relation instance $r(R)$, identify the multivariate dependencies present in the data.
\end{itemize}

This decomposition yields two complementary advantages. From a data management perspective, data ingestion does not need to adapt the data format as required by downstream dependency analysis, which is of particular value in heterogeneous data lakes. From an analytical perspective, all dependency analyses operate over a common abstraction layer, enabling a broad range of established techniques to be applied independently of the underlying domain diversity. We address the first subproblem here, and turn a relation's meaningfulness into a statistical question in Section~\ref{sec:cohesion}.

We define a predicate $P_i \colon \Omega \to \{0,1\}$ as any function mapping elements of a space $\Omega$, called the predicate domain of $P_i$, to a binary outcome. For each $A_i \in R$, we denote by $P(A_i) = \{P_i, P_j, \dots\}$ a set of predicates defined entirely in terms of that attribute. The predicate space is then the union of all such sets, $P = \bigcup_{A_i \in R} P(A_i)$, and its construction is the goal of this subproblem. The first row of Table~\ref{tab:predSpaces} illustrates a predicate space whose domain is individual tuples, and demonstrates the flexibility of this framework in accommodating heterogeneous data. Categorical and binned numerical attributes are handled naturally. Beyond these, the framework allows the definition of concise, semantically meaningful predicates to support complex types such as images or text, which can be incorporated into the multivariate analysis without special treatment.

\begin{table}[]
    \centering

    \renewcommand{\arraystretch}{1.2} 
    \caption{Predicate spaces for the predicate domain of tuples $\Omega=r(R)$ and tuple pairs $\Omega=r(R)\times r(R)$.
    }

    \resizebox{\columnwidth}{!}{
    \begin{tabular}{|c|c|c|c||c|c|}
    \hline
    & \multicolumn{5}{|c|}{Attribute and Type} \\ \cline{2-6}
    & Name & Salary & Tax & Photo & Message \\ \cline{2-6}
    & String & Real & Real & Image & Text \\ \hline

    $\Omega=r$
    & \footnotesize$\begin{matrix}="Alice"\\has\ hyphen\end{matrix}$
    & \footnotesize$\begin{matrix}\in[0,8000]\\>8000\end{matrix}$
    & \footnotesize$\begin{matrix}=0\%\\>0\%\end{matrix}$
    & \footnotesize$\begin{matrix}is\ person?\\in\ color?\end{matrix}$
    & \footnotesize$\begin{matrix}size<256\\positive\end{matrix}$
    \\ \hline

    $\Omega=r\times r$
    & \footnotesize$\begin{matrix}=Name\\\ne Name\end{matrix}$
    & \footnotesize$\begin{matrix}>Salary\\<Salary\end{matrix}$
    & \footnotesize$\begin{matrix}>Tax\\<Tax\end{matrix}$
    & \footnotesize$\begin{matrix}\sim Hue\ spectrum\\\sim Embedding\end{matrix}$
    & \footnotesize$\begin{matrix}\sim Vocabulary\\\sim Embedding\end{matrix}$
    \\ \hline

    \end{tabular}
    }

    \renewcommand{\arraystretch}{1} 
    \label{tab:predSpaces}
\end{table}

The principal limitation of projecting data into binary predicates is the inevitable loss of information, making the approach unsuitable for tasks requiring a detailed characterization of dependencies. However, since the framework aims to determine whether a dependency exists rather than quantify its nature, this trade-off is acceptable. Moreover, the abstraction extends beyond individual tuples. The second row of Table~\ref{tab:predSpaces} shows how it generalizes to pairs of tuples, yielding new spaces that have proven useful in other fields~\cite{chu2013discovering,pena2019discovery}. Our implementation instantiates the pair-of-tuples domain, whereas the non-tabular predicates in Table~\ref{tab:predSpaces} illustrate the broader reach of the abstraction and are left to future work.

\subsection{Meaningfulness as a hypothesis test}\label{sec:cohesion}

To determine multivariate dependencies, the framework employs a statistical test assuming complete attribute independence as the null hypothesis:
\[
  \begin{split}
    H_0\colon &\quad
      p(A_1 = a_1, A_2 = a_2, \dots, A_m = a_m)
      = \prod_{i} p(A_i = a_i), \\
    H_1\colon &\quad
      p(A_1 = a_1, A_2 = a_2, \dots, A_m = a_m)
      \neq \prod_{i} p(A_i = a_i).
  \end{split}
\]
Let $p(P_i)$ denote the probability that predicate $P_i$ is satisfied when evaluated on a uniformly drawn element from its predicate domain. Under the null hypothesis, the joint probability of any set of predicates drawn from distinct attributes satisfies the factorization property:
\begin{equation}\label{eq:jointIndep}
  p(P_i \land P_j \land \cdots) = p(P_i) \cdot p(P_j) \cdots
\end{equation}
for all $P_i \in P(A_x),\, P_j \in P(A_y), \dots$ with $A_x, A_y,
\dots$ pairwise distinct. 
This implies that for any set of predicates $P'\subseteq P$, we have conditional independence, where probabilities are not affected by conditioning on fewer predicates:
\begin{equation}\label{eq:condIndep}
  p(P_i \mid P' \setminus \{P_i\}) = p(P_i \mid Q \setminus \{P_i\})
\end{equation}
for all $P_i \in P'$ and all $Q \subset P'$ with $P_i \in Q$.

The goal of this subproblem is to enumerate sets of predicates that deviate from this independence, thus rejecting the null hypothesis and implying the existence of a relationship between their respective attributes. This is achieved by statistically modeling conditional distribution terms like $p(P_i \mid P_j)$ and $p(P_i)$, and evaluating the likelihood of their equality under the null hypothesis. Two propositions ground this reasoning.

\begin{proposition}[Soundness of the abstraction]\label{prop:soundness}
Let $P_1,\dots,P_l \in P$ be predicates over distinct attributes $A_{x_1},\dots,A_{x_l}$. If $p(P_1 \land \dots \land P_l) \neq \prod_i p(P_i)$, then $A_{x_1},\dots,A_{x_l}$ are not mutually independent. The converse does not hold.
\end{proposition}
\begin{proof}[Proof]
Since each $P_i$ is derived strictly from its corresponding attribute $A_{x_i}$, independence among the attributes guarantees independence among the predicates. Therefore, if the predicates' joint probability fails to factorize (i.e., it does not equal the product of their individual probabilities), the underlying attributes must be dependent. The converse does not hold, because predicates are lossy abstractions and, thus, if an attribute dependency is entirely contained within the values mapped to a single predicate outcome (e.g., variations inside a single numerical bin), the binary predicates will not detect it, leaving their probabilities perfectly independent.
\end{proof}

\begin{proposition}[Equivalence with total correlation]\label{prop:tc}
Let the set $P' = \{P_1,\dots,P_l\}$ be predicates over pairwise-distinct attributes, and let $TC(P') = \sum_i H(P_i) - H(P_1,\dots,P_l)$ denote the total correlation of the induced Bernoulli vector. Then Equation~\eqref{eq:condIndep} holds for every subset of $P'$ iff $TC(P') = 0$.
\end{proposition}
\begin{proof}[Proof]
Total correlation ($TC$) measures the amount of shared information among a set of variables. If $TC(P') = 0$, the predicates are perfectly independent, so knowing the probability of some predicates provides no information about the others. That is, every conditional probability equals its marginal probability, which directly satisfies Equation~\eqref{eq:condIndep}. Conversely, suppose Equation~\eqref{eq:condIndep} holds, meaning that conditioning a predicate on other sets never changes its probability. Expanding the joint probability step-by-step via the chain rule, $p(P_1,\dots,P_l) = p(P_1) \cdot p(P_2 \mid P_1) \cdots p(P_l \mid P_1,\dots,P_{l-1})$, we can strip the conditions from every term (e.g., $p(P_2 \mid P_1)$ becomes $p(P_2)$), since Equation~\eqref{eq:condIndep} guarantees context does not matter. The chain thus collapses into a simple product of individual probabilities, so the joint distribution factorizes perfectly and the variables are entirely independent, so $TC(P') = 0$.
\end{proof}

This statistical foundation allows dependency decisions to be made by evaluating predicates on domain samples, and ensures uncertainty is properly accounted for in the modeled distributions. All conclusions about attribute dependencies are grounded in a principled statistical test, quantifying how unlikely the observed predicate behavior would be under true attribute independence. Crucially, this reasoning operates entirely within the predicate space, remaining agnostic to the heterogeneous data types of the source attributes. Remaining at this level of abstraction does not pose a problem, as, by Proposition~\ref{prop:soundness}, any dependency detected in the predicate space guarantees a dependency among the source attributes, regardless of their types.

\subsection{Problem statement}\label{sec:problem_statement}

To formalize the data composition problem, we first define a pairwise relationship between attributes as follows.

\begin{definition}[Related attributes]\label{def:related}
Attributes $A_i$ and $A_j$ are \emph{related}, denoted $A_i \sim A_j$, iff there exists a predicate set $P' \subseteq P$, over pairwise-distinct attributes, that violates conditional independence (Equation~\eqref{eq:condIndep}), with $P(A_i) \cap P' \neq \emptyset$ and $P(A_j) \cap P' \neq \emptyset$.
\end{definition}

In other words, a predicate set showcasing conditional dependence establishes a relationship between every pair of attributes represented within it. In a lake, predicates are evaluated over the instance obtained by joining the input relations along the candidate keys. Data composition is then stated as follows.

\begin{problem}[Data Composition]\label{prob:composition}
Given relation instances $r(R_1),\dots,$ $r(R_k)$, a set of join candidates $J$, and a predicate space $P$ over their attributes, compute the composition $\mathcal{C}=\{C_1,\dots,C_q\}$ of the non-key attributes induced by the transitive closure of $\sim$, materializing each block as a relation and replicating every join key of $J$ needed to associate those attributes.
\end{problem}

That is, the output is a graph where nodes indicate attributes and edges identify meaningful dependencies between attributes, thus generating partitions (or components) of semantically related attributes (we further extend this in Section~\ref{sec:pipeline}). Unrelated attributes form their own relation rather than being discarded, making the composition complete (R3), and the replicated keys keep the resulting relations mutually joinable. Section~\ref{sec:scalable_algorithm} develops the algorithm that decides $\sim$ statistically.

\section{Scalable algorithm} \label{sec:scalable_algorithm}

Thanks to the predicate abstraction, the problem of discovering multivariate attribute associations can be reduced to well-studied problems such as frequent itemset mining and association rule mining. Algorithms such as Apriori~\cite{DBLP:conf/sigmod/AgrawalIS93} or FP-Growth~\cite{DBLP:conf/sigmod/HanPY00} may be straightforwardly adapted to discover predicate sets that deviate from independence, and by extension to identify the existence of multivariate attribute dependencies. However, while this flexibility is one of the principal advantages of the predicate abstraction, a dedicated algorithm designed specifically for this problem can achieve significantly better performance.

Let $P = \{P_1, P_2, \dots, P_k\}$ be the predicate space for a predicate domain $\Omega$, where $P_i \colon \Omega \to \{0,1\}  \ \ \forall P_i \in P$. The space of predicate sets can be represented as a partially ordered set under the subset relation $\subset$, as in Figure~\ref{fig:lattice}. In this representation, each node corresponds to a predicate set $\{P_i, P_j, \dots\}$ and is associated with its joint predicate probability $p(P_i, P_j, \dots)$. Each edge corresponds to the addition of a predicate, $\{P_j, \dots\} \rightarrow \{P_i, P_j, \dots\}$, and is associated with the conditional probability $p(P_i \mid P_j, \dots)$. Under this representation, identifying conditional probabilities that violate Equation~\ref{eq:condIndep} can be translated to finding nodes $N\subset P$ where the conditional probability of the edge adding $P_i$, $p(P_i|N)$, is different from the one  from the edge adding $P_i$ to any subset $N'\subset N$, $p(P_i|N')$.

To discover all such predicate sets, one could adopt search strategies analogous to those of classical association rule mining: design an efficient procedure to evaluate each $P_i$ over all elements of $\Omega$, empirically estimate each conditional probability $p(P_i \mid P_j, \dots)$, and accept those predicate sets whose conditional probabilities differ sufficiently from the independence baseline. However, such approaches face two principal scalability challenges:

\begin{itemize}[leftmargin=1em, itemsep=0pt]
    \item \textbf{Row scalability}: Regardless of the algorithm design, every element of the domain must be visited at least once. This is a significant constraint for large datasets, particularly given the increased domain size of more complex predicate domains.
    \item \textbf{Column scalability}: The predicate space grows linearly with the number of attributes. Thus, the number of candidate predicate sets is exponential in the number of attributes, posing severe limitations for dependency discovery in large attribute sets.
\end{itemize}

Both challenges have well-established solutions in the association rule mining literature that transfer naturally to this setting. Row scalability can be improved through sampling and estimating conditional probabilities from a subset $\Omega' \subseteq \Omega$~\cite{AR_sample_riondato2014efficient,DBLP:conf/vldb/Toivonen96}. Column scalability can be addressed by imposing minimum support or maximum predicate-set size thresholds to prune the search space~\cite{DBLP:conf/sigmod/AgrawalIS93,DBLP:conf/sigmod/HanPY00}. However, these approaches introduce arbitrary hyperparameters (e.g., sampling delimiters or support thresholds), which are undesirable due to the lack of a generalizable and mathematically validated criterion to instantiate them. In this section, we present the core algorithm of \sys, which adapts these methods to discover independence-breaking predicate sets with minimal computation and without arbitrary hyperparameters (Section~\ref{sec:posteriors}).

\subsection{Posterior modeling}\label{sec:posteriors}

Instead of relying on point estimates of probabilities computed from a single, arbitrarily sized subset $\Omega' \subseteq \Omega$, we model the full distribution of each probability. These distributions are inferred from dynamically sized samples $\Omega(P_i,P_j,\dots)\subseteq\Omega$, chosen independently for each predicate set. Treating the evaluation of a predicate on a randomly sampled element of $\Omega$ as a Bernoulli random variable, the corresponding population probability is modeled using Bayesian Learning with a Beta posterior,
\[
p(P_i,P_j,\dots)\sim\mathrm{Beta}(\alpha+1,\beta+1),
\]
where
\[
\begin{gathered}
\alpha =
\left|
\left\{
t\in\Omega(P_i,P_j,\dots)
\,\middle|\,
P_i(t)\land P_j(t)\land\cdots
\right\}
\right|, \\
\beta =
|\Omega(P_i,P_j,\dots)|-\alpha.
\end{gathered}
\]
Throughout, we write $\mathrm{attr}(P_i)$ for the attribute over which $P_i$ is defined, and $\mathrm{attrs}(P') = \{\mathrm{attr}(P_i) : P_i \in P'\}$.  Violations of Equation~\eqref{eq:condIndep} are determined statistically. Given two Beta posteriors with means $\mu_1,\mu_2$ and variances $\sigma_1^2,\sigma_2^2$ for the two conditional probabilities compared by Equation~\eqref{eq:condIndep}, we define the \textit{soundness} of the comparison as $s \;=\; \frac{|\mu_1 - \mu_2|}{\sqrt{\sigma_1^2 + \sigma_2^2}}$, the number of pooled standard deviations separating the estimates\footnote{For numerical stability, instead of directly modeling the beta distributions, we model the logarithm of their odds ratio.} \cite{martin_how_2025,martin_discovering_2025}. The precision of every estimate is driven by a user-defined threshold $\epsilon$, the \emph{approximation factor}. Each posterior is refined (Section~\ref{sec:refinement}) until its variance falls below $\epsilon$. A comparison is accepted as a violation when $s \geq s_0$, set to $s_0=\Phi^{-1}(1-0.05)$ throughout\footnote{$\Phi^{-1}$ denotes the inverse cumulative normal distribution. A violation indicates the two distributions would fail a statistical equality test of $95\%$ confidence.}. Violating comparisons are the backbone of the algorithm, serving as a guide for exploration (Section~\ref{sec:expansion}) and as a baseline to determine resulting attribute dependencies (Section~\ref{sec:pipeline}).


Rather than traversing the lattice using breadth-first search (as in Apriori) or depth-first search (as in FP-Growth) over the full domain, the proposed algorithm uses these violations to maintain an incrementally growing set of explored predicate sets, denoted by $N$.  For every explored predicate set $P'\in N$, the algorithm stores the parameters of its probability distribution, $\alpha(P')$ and $\beta(P')$, inferred from evaluating all elements in its current sample $\Omega(P')$. Initially, the explored set contains only the root, all singletons, and all
cross-attribute pairs,
\[
\begin{split}
N=\{\emptyset\}
&\cup \bigl\{\{P_i\} : P_i \in P\bigr\}\\
&\cup \bigl\{\{P_i,P_j\} : \mathrm{attr}(P_i) \neq
\mathrm{attr}(P_j)\bigr\};
\end{split}
\]

The explored set is then iteratively updated through two stages:

\begin{itemize}[leftmargin=1em, itemsep=0pt]
    \item \textbf{Distribution refinement} (Section~\ref{sec:refinement}). For every explored predicate set $P'\in N$, its sample is expanded, $\Omega(P')\gets\Omega(P')\cup\Delta\Omega(P'),$
    and the corresponding parameters $\alpha(P')$ and $\beta(P')$ are updated, reducing the uncertainty of the estimated probability.
    \item \textbf{Lattice expansion} (Section~\ref{sec:expansion}). The boundary of the explored lattice is enlarged whenever the refined distributions provide sufficient statistical evidence that unexplored neighboring predicate sets may violate Equation~\ref{eq:condIndep}.
\end{itemize}

\begin{figure}[]
\centering
\Description{Lattice that exemplifies the search developed by \sys, encompassing the exploration of paths whose joint dependence deviates from statistical independence.}
\begin{tikzpicture}[
  every node/.style={font=\tiny},
  n/.style={draw=black!70, rounded corners=1pt, fill=black!5,
    inner sep=1.6pt, align=center},
  rnode/.style={n, draw=red!75!black, thick, fill=red!6},
  u/.style={draw=black!35, dashed, rounded corners=1pt,
    inner sep=1.6pt, align=center, text=black!45},
  fe/.style={draw=black!30, thin},
  se/.style={draw=blue!65!black, thick},
  he/.style={draw=red!75!black, very thick},
  pe/.style={draw=black!20, densely dashed, thin},
  lab/.style={fill=white, inner sep=0.4pt, font=\tiny}
]
\node[n] (root) at (0,0) {$\emptyset$};
\node[n] (F) at (-3.15,1.08) {$\{P_F\}$\\$.200$};
\node[n] (R) at (-1.05,1.08) {$\{P_R\}$\\$.250$};
\node[n] (V) at (1.05,1.08) {$\{P_V\}$\\$.202$};
\node[n] (L) at (3.15,1.08) {$\{P_L\}$\\$.250$};
\node[n]     (FR) at (-3.55,2.36) {$\{P_F,P_R\}$\\$.050$};
\node[rnode] (FV) at (-2.13,2.36) {$\{P_F,P_V\}$\\$.045$};
\node[rnode] (RV) at (-0.71,2.36) {$\{P_R,P_V\}$\\$.058$};
\node[n]     (FL) at (0.71,2.36)  {$\{P_F,P_L\}$\\$.050$};
\node[n]     (RL) at (2.13,2.36)  {$\{P_R,P_L\}$\\$.063$};
\node[n]     (VL) at (3.55,2.36)  {$\{P_V,P_L\}$\\$.051$};
\node[rnode] (FRV) at (-2.4,3.64) {$\{P_F,P_R,P_V\}$\\$.030$};
\node[u] (FRL) at (-0.15,3.64) {$\{P_F,P_R,P_L\}$};
\node[u] (FVL) at (1.75,3.64) {$\{P_F,P_V,P_L\}$};
\node[u] (RVL) at (3.55,3.64) {$\{P_R,P_V,P_L\}$};
\node[u] (top) at (0.6,4.64) {$\{P_F,P_R,P_V,P_L\}$};
\draw[fe] (root)--(F); \draw[fe] (root)--(R);
\draw[se] (root)--(V) node[lab,pos=0.5]{$p(P_V) = .202$}; \draw[fe] (root)--(L);
\draw[fe] (F)--(FR); 
\draw[se] (F)--(FV) node[lab,pos=0.5]{$p(P_V|P_F) = .225$};
\draw[se] (R)--(RV) node[lab,pos=0.5]{$p(P_V|P_R) = .232$};
\draw[fe] (R)--(FR);
\draw[fe] (V)--(FV); \draw[fe] (V)--(RV);
\draw[fe] (F)--(FL); \draw[fe] (R)--(RL); \draw[fe] (V)--(VL);
\draw[fe] (L)--(FL); \draw[fe] (L)--(RL); \draw[fe] (L)--(VL);
\draw[he] (FR)--(FRV) node[lab,pos=0.5]{$p(P_V|P_F \wedge P_R)=.600$};
\draw[fe] (FV)--(FRV); \draw[fe] (RV)--(FRV);
\draw[pe] (FR)--(FRL); \draw[pe] (FL)--(FRL); \draw[pe] (RL)--(FRL);
\draw[pe] (FV)--(FVL); \draw[pe] (FL)--(FVL); \draw[pe] (VL)--(FVL);
\draw[pe] (RV)--(RVL); \draw[pe] (RL)--(RVL); \draw[pe] (VL)--(RVL);
\draw[pe] (FRV)--(top); \draw[pe] (FRL)--(top);
\draw[pe] (FVL)--(top); \draw[pe] (RVL)--(top);
\node[align=left, text=black!60] at (-3.2,0.2)
  {$P_F \to \textit{floor\_number} = 0$ \\
  $P_R \to \textit{flood\_risk} > 0.3$\\
  };
\node[align=left, text=black!60] at (2.8,0.2)
  {$P_V \to \textit{market\_value} < 250.000$ \\
  $P_L \to \textit{locker\_number} < 50$\\
  };
\node[align=left, text=black!45] at (3.1,4.40)
  {\textit{locker\_number} is unrelated,\\ so its paths are pruned};
\end{tikzpicture}
\caption{Snippet of the lattice for the example of Figure~\ref{fig:discovery_example}. Nodes show the posterior mean of the joint probability at termination and edge labels give the conditional probability of the added predicate. The weak conditional signals $p(P_V\!\mid\!P_F)=.225$ and $p(P_V\!\mid\!P_R)=.232$
deviate from $p(P_V)=.202$ just beyond the effect-size floor at $\epsilon=10^{-5}$ ($\delta_{\min}\approx.007$), relating \textit{floor\_number} and \textit{flood\_risk} to \textit{market\_value}. Expansion then reaches $\{P_F,P_R,P_V\}$, where the conditional $p(P_V \mid P_F \wedge P_R) = .600$ confirms the strong joint dependence ($s\gg s_0$). \textit{locker\_number} relates to nothing, so its supersets are pruned. 
}
\label{fig:lattice}
\end{figure}

\subsection{Distribution Refinement}\label{sec:refinement}

Given the current explored set $N$, where each predicate set $P' \in N$ has posterior parameters $\alpha(P')$ and $\beta(P')$ inferred from sample $\Omega(P')$, the goal is to enlarge each sample, $\Omega(P') \gets \Omega(P') \cup \Delta\Omega(P')$, and update the corresponding posterior parameters without explicitly storing $\Omega(P')$, which is prohibitively memory-intensive.

The algorithm maintains the invariant that all variances remain approximately equal. Equivalently, distributions are refined at a rate proportional to their current uncertainty, preventing computation from being wasted on already well-saturated probabilities while focusing effort on poorly estimated ones. At each iteration, the algorithm identifies all explored predicate sets whose posterior variance still exceeds the approximation factor $\epsilon$. If no such predicate sets remain, refinement is complete for the current explored set. The algorithm terminates once lattice expansion (Section~\ref{sec:expansion}) also produces no new candidates (Section~\ref{sec:termination}).

A new domain sample $\Delta\Omega\subseteq\Omega$ is then generated.\footnote{The sample size starts at $1000$ elements and increases exponentially up to a maximum of $10^6$ in our implementation. This allows the algorithm to quickly identify distributions that converge early while avoiding unnecessary computation.} This sample is propagated through the explored lattice by filtering it along each edge, retaining only those elements satisfying the predicates encountered along the path\footnote{There are $|P'|!$ paths from the root to node $P'$. The algorithm uses the inferred distributions to estimate the optimal path parent of each node.}. Consequently, when visiting a predicate set $P'$, the algorithm updates $\alpha(P')$ and $\beta(P')$ using the number of sampled elements that satisfy, or fail to satisfy, all predicates in $P'$.

\subsection{Lattice Expansion}\label{sec:expansion}

Given the current explored set $N$ together with the posterior parameters of every explored predicate set, the goal of this stage is to expand $N\gets N\cup\Delta N.$
The primary challenge is to prune the exponentially large search space without discarding predicate sets that may reveal meaningful dependencies.

At any iteration, the posterior distributions provide sufficient information to determine whether there is statistical evidence that a predicate set violates Equation~\ref{eq:condIndep}. Let $S\subseteq N$ denote the collection of explored predicate sets containing a comparison whose soundness reaches $s_0$. Recall from Definition~\ref{def:related} that such a violating set relates every pair of attributes represented within it. The expansion set $\Delta N$ is constructed by taking each explored predicate set $P'\in N$ and adding one predicate from $P(A_i)$, for $A_i \notin \mathrm{attrs}(P')$, whenever attribute $A_i$ is related to at least one attribute already represented in $P'$. Intuitively, supersets obtained by introducing predicates from attributes that have never exhibited any statistical relationship with the attributes already present in $P'$ are not explored. Since combinations of independent attributes constitute the overwhelming majority of the lattice, this pruning strategy dramatically reduces the search space while ensuring that any discovered dependency can trigger exploration of the corresponding region of the lattice.




\begin{assumption}[Subset Traceability]\label{as:seeded}
Every multivariate dependency of interest exhibits, on at least one of its proper subsets, a conditional shift of at least $\delta_{\min} = s_0\sqrt{2\epsilon}$.
\end{assumption}

This assumes that multivariate dependencies leave a statistical trail in their smaller subsets to guide the search. In the running example (Figure~\ref{fig:lattice}), flood-prone districts are slightly cheaper on average, and ground floors slightly discounted overall. Such weak marginal signals, once accepted, trigger exploration of the joint region where the dependence is strong. A traditional feature selector might not consider the weak pairwise conditional dependencies worth keeping, thus missing the stronger multivariate interaction. Consequently, \sys~implements a scalable heuristic rather than an exhaustive search: dependencies whose subsets carry no detectable signal (e.g., purely interactive ones) may be missed. Its pairwise guarantee is unaffected, as all cross-attribute pairs are explored by construction.

\subsection{Termination and cost}\label{sec:termination}

Upon termination, every explored predicate set is associated with a posterior probability distribution whose uncertainty is bounded by the user-specified approximation factor~$\epsilon$. These posterior distributions are then used to identify predicate sets that violate the conditional independence criterion of Equation~\ref{eq:condIndep}, as described previously. Consequently, the computational cost is governed primarily by the uncertainty bound $\epsilon$ and the number of dependencies present in the data, rather than by exhaustive traversal of the entire lattice, making the approach scalable to large, high-dimensional datasets.

\section{The \sys~system} \label{sec:pipeline}

In this section, we present \sys~(\textit{Discovering Intrinsic Attribute Dependencies At scale}) and its implementation details. For each relation in the input data lake, \sys~constructs a predicate space over its schema, applies the mining algorithm of Section~\ref{sec:scalable_algorithm} to obtain the set of statistically meaningful attribute dependencies ($R$), and transforms $R$ into the relations of the composed layer.

\medskip

\stitle{Input} \sys~takes as input a wide relation, typically obtained by joining a base relation with candidate datasets returned by a discovery system~\cite{fernandez2018seeping, nargesian2019data, DBLP:journals/tkde/MaynouNPFRQ26}. As argued in Section~\ref{sec:related_work}, discovery establishes that datasets \emph{can} be combined, but not whether their attributes \textit{should belong} in the same relation, a gap that \sys~addresses. Therefore, \sys~is agnostic to the discovery algorithm and join conditions used. Nonetheless, recall that predicate evaluation is modeled as a Bernoulli trial over a uniformly sampled domain $\Omega$, with the sampling distribution matching the one under evaluation. Thus, we fixe $\Omega = r(R_b)$ and require every candidate to be folded in without altering that domain. This requirement is violated by one-to-many joins, as a base tuple matching $m$ candidate tuples appears $m$ times in the result, contributing $m$ draws instead of one and reweighting the empirical distribution by match multiplicity. The resulting dependencies are indistinguishable from genuine data dependencies within the predicate space, invalidating Proposition~\ref{prop:soundness}. Each candidate should therefore be pre-aggregated by grouping on its join attribute before being combined via a multi-way left join, keeping $|\Omega| = |r(R_b)|$ regardless of candidate cardinality.

\medskip

\stitle{Instantiating the predicate space} As described in Section~\ref{sec:multivariate_dependencies}, there are numerous ways of constructing predicate spaces. The most natural approach defines predicates over the domain of individual tuples, where each predicate is satisfied whenever a given attribute assumes a particular value or falls within a specified range (first row of Table~\ref{tab:predSpaces}). While this representation yields a highly expressive predicate space, its size grows proportionally with the number of distinct values present in the data, making it impractical for large and heterogeneous datasets. Instead, we construct predicates over the domain of tuple pairs (second row of Table~\ref{tab:predSpaces}). For each attribute $A_i$, the generated predicate set is fixed:
\[
P(A_i)=
\begin{cases}
\{=A_i,\ne A_i,<A_i,>A_i\}, & \text{if } A_i \text{ has an ordered domain},\\
\{=A_i,\ne A_i\}, & \text{otherwise}.
\end{cases}
\]
Although the resulting predicate spaces are less expressive than those obtained by enumerating all distinct attribute values, they provide a practical representation for identifying relationships between attributes. By capturing only the fundamental comparison operators, they provide a compact yet informative representation of attribute behavior, substantially reducing the search space while preserving the information required for dependency discovery.

\medskip

\stitle{Attribute Dependencies} The system employs the algorithm described in Section~\ref{sec:scalable_algorithm} to infer joint predicate probability distributions whose posterior variance is bounded by the user-specified approximation factor~$\epsilon$. These distributions are subsequently used to identify the collection of predicate sets $P' \subset P$ that violate the conditional independence criterion introduced in Section~\ref{sec:multivariate_dependencies}.

Rather than directly reporting the attribute sets corresponding to statistically significant predicate combinations, the system derives a pairwise dependency score for every pair of attributes. For any two attributes $A_x$ and $A_y$, this score is defined as the maximum soundness observed among all pairs of conditional probabilities $p(P_i \mid P_j, P_k, \ldots)$ and $p(P_i \mid P_k, \ldots)$ violating Equation~\ref{eq:condIndep}, where $P_i \in P(A_x)$ and $P_j \in P(A_y)$. Intuitively, the score measures the greatest change in the behavior of a predicate derived from $A_x$ resulting from conditioning on a predicate derived from $A_y$. Consequently, it quantifies the strongest statistical evidence that $A_y$ influences the behavior of $A_x$, and can therefore be used directly as the test statistic for the independence hypothesis test described in Section~\ref{sec:multivariate_dependencies}.

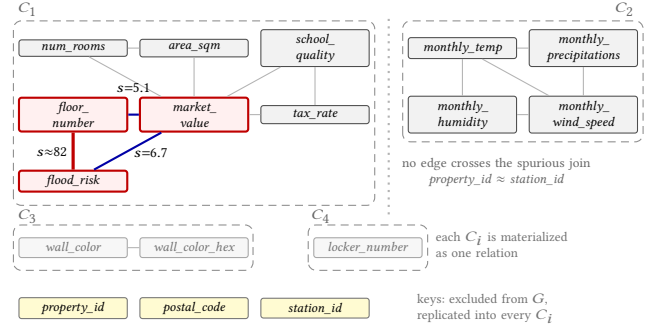
\begin{figure}[]
\centering
\Description{The association graph generated by \sys, with the resulting relations after the data has been composed. Each relation indicates their attributes via nodes and edges indicate the attributes that showcase statistical independence.}
\begin{tikzpicture}[
  every node/.style={font=\tiny},
  n/.style={draw=black!70, rounded corners=1pt, fill=black!5,
    inner sep=1.6pt, align=center, text width=1.32cm},
  rnode/.style={n, draw=red!75!black, thick, fill=red!6},
  inert/.style={n, draw=black!35, fill=black!3, text=black!55},
  k/.style={n, draw=black!55, fill=yellow!22},
  fe/.style={draw=black!30, thin},
  se/.style={draw=blue!65!black, thick},
  he/.style={draw=red!75!black, very thick},
  cbox/.style={draw=black!45, densely dashed, rounded corners=3pt},
  sep/.style={draw=black!30, dotted, thick},
  clab/.style={text=black!60, inner sep=1pt},
  ann/.style={text=black!55, align=left, inner sep=1pt},
  lab/.style={fill=none, text opacity=1, inner sep=0.5pt}
]
\draw[cbox] (-4.09,2.16) rectangle (0.69,4.60);
\draw[cbox] (1.06,3.00) rectangle (4.24,4.60);
\draw[cbox] (-4.09,1.30) rectangle (-0.92,1.90);
\draw[cbox] (-0.19,1.30) rectangle (1.39,1.90);
\draw[sep] (0.875,2.04) -- (0.875,4.65);
\node[n]     (NR) at (-3.30,4.24) {\textit{num\_rooms}};
\node[n]     (AS) at (-1.70,4.24) {\textit{area\_sqm}};
\node[n]     (SQ) at (-0.10,4.24) {\textit{school\_}\\\textit{quality}};
\node[rnode] (FN) at (-3.30,3.36) {\textit{floor\_}\\\textit{number}};
\node[rnode] (MV) at (-1.70,3.36) {\textit{market\_}\\\textit{value}};
\node[n]     (TR) at (-0.10,3.36) {\textit{tax\_rate}};
\node[rnode] (FR) at (-3.30,2.48) {\textit{flood\_risk}};
\draw[fe] (MV)--(AS); \draw[fe] (MV)--(NR); \draw[fe] (MV)--(SQ);
\draw[fe] (MV)--(TR); \draw[fe] (AS)--(NR); \draw[fe] (TR)--(SQ);
\draw[se] (MV)--(FN) node[lab,pos=0.5,yshift=10pt]{$s{=}5.1$};
\draw[se] (MV)--(FR) node[lab,pos=0.55,xshift=10pt]{$s{=}6.7$};
\draw[he] (FN)--(FR) node[lab,pos=0.5,xshift=-8pt]{$s{\approx}82$};
\node[n] (MT) at (1.85,4.24) {\textit{monthly\_temp}};
\node[n] (MP) at (3.45,4.24) {\textit{monthly\_}\\\textit{precipitations}};
\node[n] (MH) at (1.85,3.36) {\textit{monthly\_}\\\textit{humidity}};
\node[n] (MW) at (3.45,3.36) {\textit{monthly\_}\\\textit{wind\_speed}};
\draw[fe] (MT)--(MP); \draw[fe] (MT)--(MH); \draw[fe] (MP)--(MW);
\draw[fe] (MH)--(MW); \draw[fe] (MT)--(MW);
\node[inert] (WC) at (-3.30,1.60) {\textit{wall\_color}};
\node[inert] (WH) at (-1.70,1.60) {\textit{wall\_color\_hex}};
\draw[fe] (WC)--(WH);
\node[inert] (LN) at (0.60,1.60) {\textit{locker\_number}};
\node[k] (PI) at (-3.30,0.80) {\textit{property\_id}};
\node[k] (PC) at (-1.70,0.80) {\textit{postal\_code}};
\node[k] (SI) at (-0.10,0.80) {\textit{station\_id}};
\node[clab] at (-3.90,4.75) {$C_1$};
\node[clab] at (4,4.75)  {$C_2$};
\node[clab] at (-3.90,2.02) {$C_3$};
\node[clab] at (-0.02,2.02) {$C_4$};
\node[ann, text width=2.6cm] at (2.8,1.64)
  {each $C_i$ is materialized\\as one relation};
\node[ann, text width=2.8cm] at (2.65,0.80)
  {keys: excluded from $G$,\\replicated into every $C_i$};
\node[ann, align=center, text width=8.2cm, text=black!60] at (2.3,2.6)
  {no edge crosses the spurious join \\ \textit{property\_id} $\approx$ \textit{station\_id}};
\end{tikzpicture}
\caption{The association graph $G$ for the running example of Figure~\ref{fig:discovery_example}.
Edges join attributes co-occurring in some violating predicate set. We indicate the soundness of relationships breaking the independence hypothesis in Figure \ref{fig:lattice}, following the posteriors at $\epsilon=10^{-5}$.
The red edge carries the comparison $p(P_V \mid P_F, P_R)$ vs.\
$p(P_V \mid P_R)$, $|.600-.232|/\sqrt{2\epsilon} \approx 82$. The four
connected components are the composition $\mathcal{C}$ and become the
four relations of Figure~\ref{fig:discovery_example}c. 
}
\label{fig:assoc}
\end{figure}

\medskip

\stitle{Compositional Components} These pairwise scores are turned into the association graph $G$, an undirected graph in which two columns are connected if the maximum soundness between any conditional predicate meets or exceeds a fixed threshold: $\Phi^{-1}\!\left(1-\frac{0.01}{|P'|}\right)$, where $\Phi^{-1}$ denotes the inverse cumulative distribution function of the standard normal distribution and $P'$ is the set of predicates attributed to the score. This is the $(1-0.01/|P'|)$ quantile of the standard normal distribution (i.e., roughly two standard deviations, applying a Bonferroni correction across the predicates entering each score). Composing then reduces to traversing the constructed graph and defining each connected component as a relation. For example, given attributes $\{A,B,C\} \in R$, if $A$ relates to $B$ and $B$ relates to $C$, then $A$, $B$, and $C$ are placed in the same component because, even if $A$ and $C$ are not directly related, they have a semantic link through $B$. Columns with no qualifying relationship are treated as univariate noise and collected into a single component rather than discarded outright. Columns that a join was actually performed on are propagated into every resulting relation, so each table remains independently joinable back via a shared key. Figure~\ref{fig:assoc} shows the compositional components generated for the example in Figure~\ref{fig:discovery_example}.


\section{Experiments} \label{sec:experiments}

Our evaluation answers four questions, tied to the requirements defined in Section~\ref{sec:related_work}:
\begin{itemize}[leftmargin=2.6em, itemsep=1pt]
    \item[\textbf{RQ1}] Does the mining algorithm scale in rows and columns beyond the classical lattice miners? (R4)
    \item[\textbf{RQ2}] Does composition separate noise from signal better than unsupervised feature selection? (R1--R3)
    \item[\textbf{RQ3}] Does composing the data before modeling improve downstream ML pipelines? (R2--R3)
    \item[\textbf{RQ4}] Does a single, untargeted composition support inference over several targets defined only afterwards? (R1, R3)
\end{itemize}

\stitle{Benchmarks} We collect 17 public datasets spanning binary, multiclass, and regression tasks (Table~\ref{tab:benchmarks}), and derive from each two noisy benchmarks by adding columns, emulating the result of spurious joins. The \textit{univariate noise} (UN) benchmark adds, for every original non-target column, two columns of pure noise (one numerical, one categorical, with varying distributions) and two shuffled copies of the column itself, preserving its distribution while destroying every relationship. The \textit{multivariate noise} (MN) benchmark adds, on top of UN, a cluster of synthetic features correlated with one another but with nothing else, emulating an entire relation spuriously merged. This results in benchmarks ranging from 41 to 406 columns.

\begin{table}
    \caption{Characteristics of the considered datasets.}
    \centering
    \small
    \resizebox{\columnwidth}{!} {
        \begin{tabular}{cccc}
            \toprule
            \textbf{Dataset} & \textbf{Task} & \textbf{\#Rows} & \textbf{\#Columns}\\
            \midrule
            Bank & Binary & 45k & 16\\
            Default Payment & Binary & 30k & 24\\
            Jannis & Binary & 57k & 54\\
            Miniboone & Binary & 73k & 50\\
            Covertype & Multiclassification (7) & 40k & 13\\
            Drive Diagnosis & Multiclassification (10) & 58k & 48\\
            Dry Beans & Multiclassification (7) & 13k & 17\\
            Mice Protein & Multiclassification (8) & 1k & 81\\
            Pendigits & Multiclassification (10) & 11k & 17\\
            Students & Multiclassification (3) & 5k & 36\\
            Appliances & Regression & 20k & 28\\
            Avocado Sales & Regression & 18k & 13\\
            Diamonds & Regression & 54k & 9\\
            House Sales & Regression & 21k & 17\\
            Nasa & Regression & 45k & 22\\
            Pol & Regression & 15k & 48\\
            Superconductivity & Regression & 21k & 82\\
            \bottomrule
        \end{tabular}
    }
    \\
    \vspace{0.1cm}
    \label{tab:benchmarks}
\end{table}


These benchmarks mimic Figure \ref{fig:discovery_example}b: discovery produces a large dataset mixing related attributes (base datasets) with uninformative ones (all the generated columns), and the task is to recover the meaningful relations. UN requires removing the noise while keeping the original dataset, whereas MN additionally requires isolating the intra-correlated cluster. Injecting synthetic noise gives an unambiguous ground truth stating which columns must not be kept, letting us measure how much composition benefits downstream analysis. Nonetheless, we acknowledge the need to validate on naturally occurring heterogeneous lakes, which is left for future work. In particular, uncontrolled compositions in such repositories might combine all datasets into a single component. Even if supported by our statistical criterion, such a result hinders downstream usage. ~$\epsilon$ is set to its default value of $10^{-5}$ for Sections \ref{sec:lima_vs_ufs}--\ref{sec:synth_targets}.



\subsection{Scalability of the mining algorithm}\label{sec:lima_scalability}

The predicate abstraction reduces multivariate dependency discovery to the problem of frequent-itemset and association-rule mining, for which Apriori~\cite{DBLP:conf/sigmod/AgrawalIS93} and FP-Growth~\cite{DBLP:conf/sigmod/HanPY00} represent the canonical mining algorithms. We compare \sys~to them in order to isolate the effect of our pruning rather than as multivariate discovery competitors, given that, as argued in Section \ref{sec:rw-computing}, no existing method mines multivariate dependencies at lake scale. Thus, the column scalability comparison below doubles as a pruning ablation: the exhaustive miners enumerate the lattice while \sys~expands it only under statistical evidence. We analyze the three algorithms on the two scalability axes introduced in Section~\ref{sec:scalable_algorithm} (rows and columns) and report results on three of the largest benchmarks, \textit{jannis}, \textit{miniboone} and \textit{drive\_diagnosis} (Figure~\ref{fig:scalability}). To ensure a fair comparison that isolates the lattice exploration strategies from the underlying evaluation mechanisms, we adapt Apriori and FP-Growth to operate on a sampled domain $\Omega' \subset \Omega$. The sample size $|\Omega'|$ is determined directly by $\epsilon$ and is chosen as the minimum value that guarantees the uncertainty of every explored lattice node remains bounded by $\epsilon$.\footnote{For a fixed sample size $n$, the probability distribution inferred from $a$ successes and $b$ failures, where $a+b=n$, attains its maximum variance at the boundary cases ($a=0$ or $b=0$). In these cases, the variance is proportional to $\frac{1}{n}$. Consequently, by selecting $n=\frac{k}{\epsilon}$, there exists a numerical constant $k$ such that every inferred distribution has variance bounded by $\epsilon$. For the statistical tests considered in this work, numerical evaluation yields $k=5$.} This represents a best-case scenario for applying these algorithms to derive predicate distributions with bounded uncertainty, providing a fair basis for comparison. Thus, Figure~\ref{fig:scalability} reports the execution time each algorithm requires to guarantee that the variance of every explored lattice node stays within the prescribed value of $\epsilon$.

\medskip

\stitle{Row Scalability} Row scalability refers to the capacity of the algorithms to deal with larger volumes of data, determined by $\epsilon$. Figure~\ref{fig:scalability} shows FP-Growth's advantage over Apriori growing with sample size: larger samples add redundancy that FP-Growth's compressed data structures exploit effectively. For highly dimensional datasets, however, these structures eventually exceed practical memory limits. Our algorithm offers an intermediate alternative, matching FP-Growth's performance while avoiding the construction of expensive global compression structures. Instead, it concentrates computation on statistically relevant regions of the search space without requiring global data structures that summarize the entire domain.

\medskip

\stitle{Column Scalability} Column scalability is particularly important in data lakes, where joins across multiple sources produce relations with hundreds of attributes. The advantage of our exploration strategy is evident here. Whereas Apriori and FP-Growth systematically explore predicate combinations that are often mutually independent, our method postpones exploration until there is statistical evidence suggesting dependence. This selective exploration substantially reduces the number of lattice nodes visited, mitigating the exponential growth inherent to exhaustive search. Importantly, this is achieved without imposing an explicit upper bound on predicate size: the algorithm is guaranteed by construction to discover all pairwise dependencies while continuing to explore deeper levels of the lattice whenever the observed evidence warrants it.

\begin{figure}
    \centering
    \Description{Scalability plots for \sys, Apriori and FP-Growth given an increasing amount of columns and rows.}
    \includegraphics[width=\linewidth]{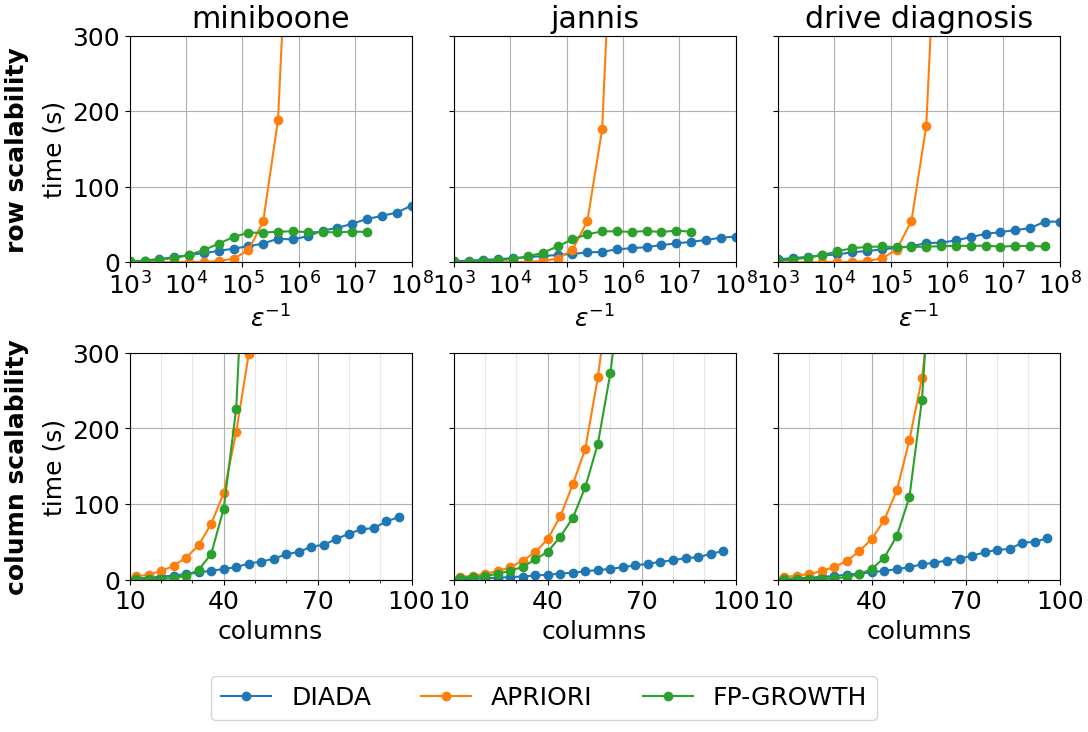}
    \caption{ Row/column scalability of \sys, Apriori and FP-Growth for estimating joint predicate probabilities of lattice predicate sets. Executions are repeated 5 times, each time evaluating on different random sets of columns, and the average is reported. Missing entries indicate memory overflows.}
    \label{fig:scalability}
\end{figure}

\subsection{Separating signal from noise}\label{sec:lima_vs_ufs}

\begin{figure*}[t]
	\begin{center}
		\includegraphics[width=1\linewidth]{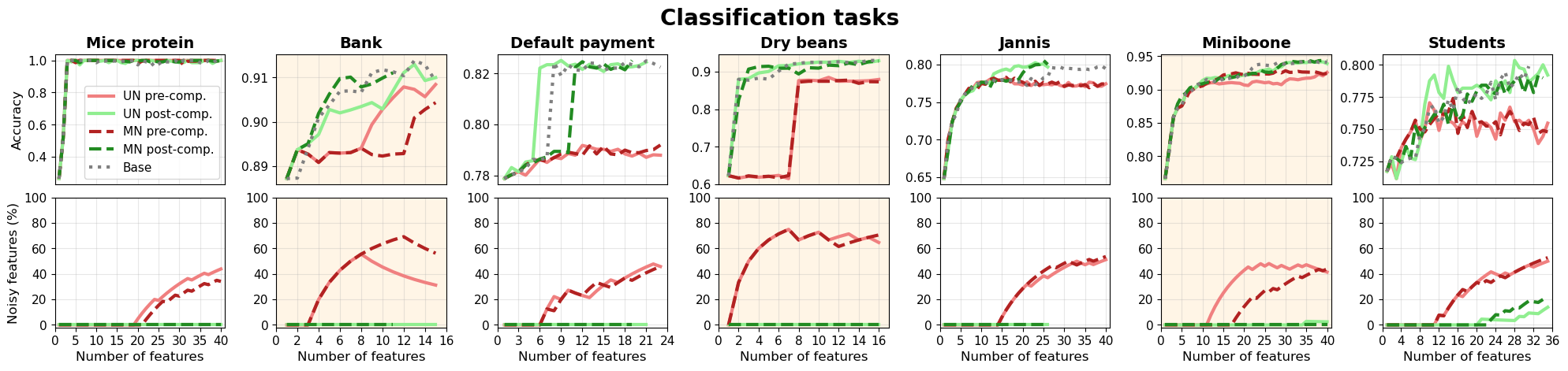}
        \includegraphics[width=1\linewidth]{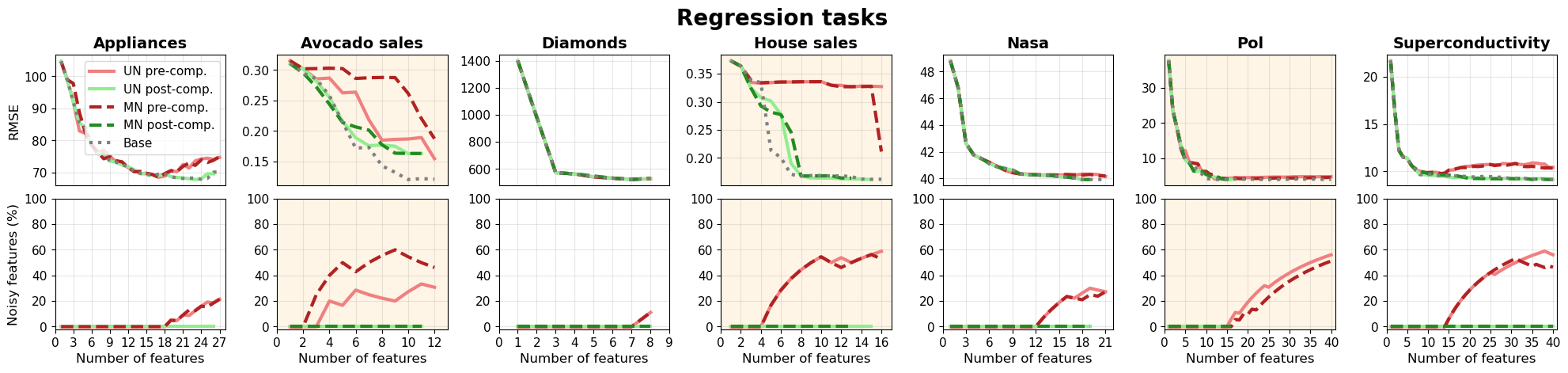}
		\caption{Model utility (top) and share of noise features selected (bottom) as ARDA selects more features, on representative benchmarks for both classification and regression tasks. \textit{Post-composition} lines stop at the width of the composed relation.}
		\label{fig:noise_removal_results}
	\end{center}
\end{figure*}

\sys~and UFS methods are the only alternatives that meet the composition requirements defined in Section \ref{sec:related_work}. Hence, we first evaluate their capacity to preserve meaningful features in the UN benchmarks. As building a ground truth from exact multivariate dependencies is unfeasible, we define meaningful features as those originally present in the benchmarks. These are popular prediction datasets, so dependencies between the target and a large subset of the attributes are assumed to be present by construction. Thus, the goal is to remove the injected noise columns while keeping the original attributes and the target. We run \sys~ and 14 popular UFS methods, classical and modern and with varied selection objectives, on all 17 benchmarks. Table~\ref{tab:feature_selection_results} presents the averaged results for the amount of selected noise columns, base attributes and targets selected across benchmarks.

\begin{table}[htbp]
    \caption{Analysis of composition quality. We report the average (± 95\% CI) percentage of noise columns, base columns and targets kept across all benchmarks.}
    \centering
    \small
    \begin{tabular}{lccc}
        \toprule
        Method & Noise (\%) & Original (\%) & Target (\%) \\
        \midrule
        \sys & \textbf{1.6 ± 2.3} & \textbf{91.5 ± 9.9} & \textbf{100.0} \\
        Cluster Dispersion \cite{dy2004feature} & 4.4 ± 6.8 & 27.3 ± 12.9 & 23.5 \\
        PFA \cite{lu2007feature} & 5.3 ± 5.4 & 40.3 ± 18.1 & 47.1 \\
        RSR \cite{zhu2015unsupervised} & 8.2 ± 9.3 & 5.0 ± 3.1 & 0.0 \\
        MCFS \cite{cai2010unsupervised} & 11.1 ± 13.8 & 12.5 ± 13.8 & 17.6 \\
        UDFS \cite{DBLP:conf/ijcai/YangSMHZ11} & 17.2 ± 10.6 & 34.2 ± 19.8 & 41.2 \\
        SPEC \cite{zhao2007spectral} & 20.6 ± 14.7 & 3.6 ± 2.4 & 5.9 \\
        NDFS \cite{li2012unsupervised} & 24.5 ± 21.0 & 3.4 ± 2.9 & 5.9 \\
        Cluster Representatives & 35.3 ± 7.7 & 13.7 ± 3.9 & 5.9 \\
        Laplacian Score \cite{he2005laplacian} & 44.2 ± 15.0 & 81.5 ± 16.1 & 88.2 \\
        AGUFS \cite{huang2021adaptive} & 62.6 ± 8.3 & 87.8 ± 13.9 & 82.4 \\
        AutoEncoder \cite{balin2019concrete} & 65.8 ± 6.2 & 13.9 ± 5.0 & 23.5 \\
        Variance Score & 66.7 ± 0.3 & 81.9 ± 8.3 & 82.4 \\
        InfFS \cite{roffo2015infinite} & 69.1 ± 14.0 & 25.6 ± 19.8 & 35.3 \\
        Correlation Uniqueness & 90.2 ± 8.3 & 7.3 ± 8.2 & 0.0 \\
        \bottomrule
    \end{tabular}
    \label{tab:feature_selection_results}
\end{table}

\sys~selects the lowest percentage of noise columns (1.6\%) while retaining the largest amount of original attributes (91.5\%). The two numbers must be read jointly: Cluster Dispersion or PFA keep noise comparatively low, but only by being so restrictive that most of the signal is discarded with it, whereas methods that retain most signal (e.g., Laplacian Score or AGUFS) admit 44--67\% noise. \sys~is also the only method that never discards the prediction target. Thus, UFS techniques are failing to preserve the only feature whose meaningfulness can be guaranteed. Two causes explain this behavior. First, UFS approaches assume the datasets' structure is worth preserving, so when the input's structure is contaminated by noise, the objective protects exactly what should be removed. Second, excessive parametrization: every method requires an assortment of hyperparameters with no principled criterion to set them, so performance varies considerably across benchmarks, as the wide confidence intervals showcase. \sys~avoids both problems by employing a statistical approach to assess multivariate attribute dependencies and using a single tuning parameter. As a result of preserving too much noise and removing too many meaningful attributes, UFS methods are not suitable for composition tasks and thus we only analyze \sys~further.

\subsection{Improving downstream pipelines}\label{sec:removing_noise}
The main goal of composition is to provide a series of meaningful relations that can benefit the performance of downstream tasks. In this experiment, we compare the results of ML modeling pipelines on the designed benchmarks (UN and MN) before (\textit{pre-}) and after (\textit{post-}) composition. To do so, we have defined a two-step modeling pipeline including a data augmentation system, ARDA~\cite{DBLP:journals/pvldb/ChepurkoMZFKK20}, to rank features by inferred relevance, and AutoML to train the models. Hence, we simulate a large-scale pipeline requiring automated processing due to the infeasibility of manually assessing feature relevance. We compare three inputs: (1) \textit{pre-composition}, where the UN and MN benchmarks reach the pipeline directly (i.e., the output of automated data discovery); (2) \textit{post-composition}, where the UN and MN benchmarks have been composed and, thus, only meaningful attributes are fed to the ML pipeline; and (3) \textit{base}, which represents the original datasets without any noise, serving as a baseline for the expected models' behavior. Note that for UN benchmarks the models ingest attributes in \textit{all} found components, whereas for MN they ingest \textit{only} the component containing the target variable. This showcases the difference between \textit{only} preserving statistically dependent features and \textit{further} composing the data into components. We evaluate model utility for an increasing amount of features selected ($k$), from 1 up to the minimum between the original width and 40. 
For representative benchmarks (Figure~\ref{fig:noise_removal_results}), we track model utility (top panels) and the percentage of noise columns selected (bottom panels). The remaining datasets, in our artifact repository, showcase the same trends.

The \textit{pre-composition} results show that without composition downstream tasks are unable to filter noise. Several benchmarks (those with a difficult target to predict) end up selecting up to 60\% of noisy features. This is a limitation of scalable methods to select features, as they rely on cheap, high-level proxies to assess relevance rather than on complex, but accurate techniques. This triggers widely studied issues in model creation (e.g., true predictors being collectively mimicked by noise variables when the marginal target signal is weak \cite{fan2008sure, fan2018discoveries}), so the resulting models rest on spurious associations. Even at high utility they generalize poorly under distribution shift, as the captured patterns do not describe replicable phenomena.

Composing first almost eliminates the problem. Consistently with Section~\ref{sec:lima_vs_ufs}, \textit{post-composition} models are built almost exclusively from original features, with the only notable exception being \textit{students}, likely due to the poor quality of its attributes that facilitates the generation of relationships between noise and original features. The utility for \textit{post-composition} models is never worse than \textit{pre-composition} models (e.g., \textit{diamonds}, \textit{nasa}) and often considerably better (e.g., \textit{default payment}, \textit{house sales}, \textit{students}). Moreover, the UN--MN gap on \textit{post-composition} models is generally minimal, confirming that the injected intra-correlated cluster is isolated and removed from modeling. In fact, some MN compositions keep fewer original features than the UN ones (e.g., \textit{bank}, \textit{default\_payment}) at no cost in utility, implying the existence of attribute groups of the original data that are internally related but independent of the target's relation. These are correctly filtered out. This is only visible in the MN benchmarks, where the ML model is fed solely the component that has the target. The only failed model is \textit{appliances} for the MN benchmark, whose composition leaves a single companion attribute in the target's relation, crippling the model. This is the result of the small percentage of error tolerated by statistical approaches, which can be addressed by lowering $\epsilon$. Overall, composition increases model performance, robustness and favors the detection of real patterns.

\begin{figure*}[t]
	\begin{center}
		\includegraphics[width=1\linewidth]{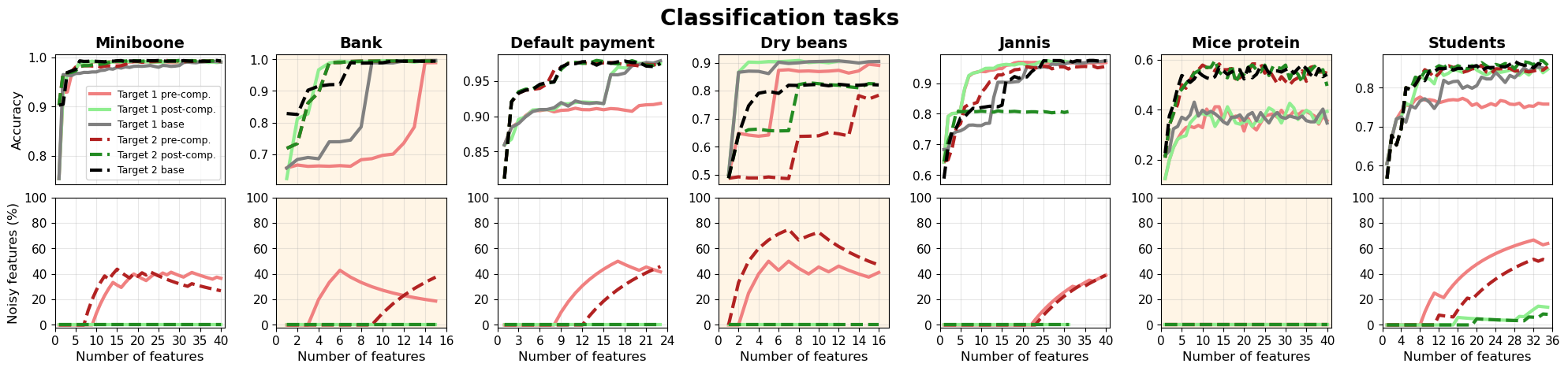}
        \includegraphics[width=1\linewidth]{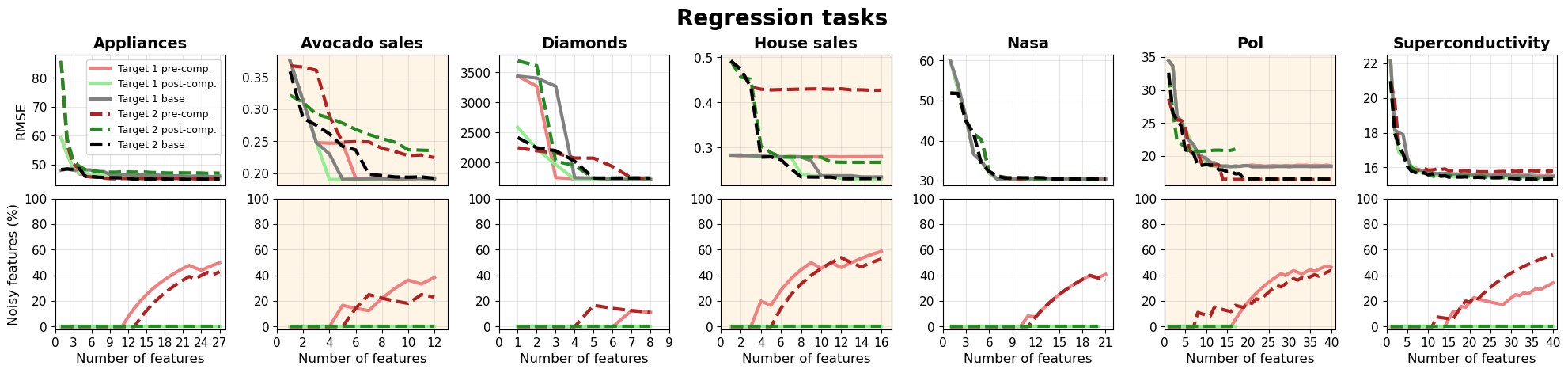}
		\caption{Model utility (top of each panel) and share of noise features selected (bottom) for two synthetic targets defined after composition, on the representative benchmarks of Figure~\ref{fig:noise_removal_results}. Remaining benchmarks in the artifact repository.}
		\label{fig:synth_targets_results}
	\end{center}
\end{figure*}

Against \textit{base}, \textit{post-composition} models perform, generally, to an equal degree, indicating that \sys~removes all the noise that negatively impacts the model. Moreover, \textit{post-composition} models perform better in some cases. On \textit{jannis} both compositions retain 26 of 54 features, and the model reaches 0.8 accuracy around 20 selected features whereas \textit{base} needs 26, with no gains beyond. Thus, composition is removing original features that provide no inference benefit \textit{without} having to rely on pairwise comparison with the target variable. \textit{pol} behaves similarly, and \textit{default payment} and \textit{house sales} do so to a lesser extent. The effect is not universal (\textit{superconductivity} and \textit{nasa} retain their full width), indicating either genuinely meaningful attributes (although redundant for the model) or redundancy that \sys~does not capture.

\subsection{One composition, any target}\label{sec:synth_targets}

To test whether one composition can serve several downstream tasks we generate two synthetic targets for each of the 17 UN benchmarks, apply a single composition process and check both models' behavior. The synthetic targets have been generated through a latent score drawn as a random linear combination over one half of the attributes for the first target and the other half for the second. This score is rank-transformed to $[0,1]$ for robustness and mapped to the appropriate task (e.g. sigmoid-thresholded for binary labels). Synthetic targets are, thus, learnable from the source attributes while distinct from the original labels. We repeat the protocol of Section~\ref{sec:removing_noise}, now for both targets, presenting \textit{pre-composition}, \textit{post-composition} and \textit{base} results in Figure~\ref{fig:synth_targets_results}.


As in Section \ref{sec:removing_noise}, \textit{post-composition} models perform close to \textit{base} models for both targets across most benchmarks, and, again, better than \textit{pre-composition} models, as the only noticeable exception is the second target of \textit{jannis}, which again exemplifies the small rate of error associated with statistical processes. Once more, noise is almost eradicated from the relations, thus favoring patterns that generalize to unseen data and not based on spurious associations. The composition was executed without knowledge of either target, implying that \sys~selects the attributes that characterize the relation, not the ones that explain a designated variable. Hence, a single composition process generates relations that can serve meaningful attributes to any subsequent analytical task.

\section{Conclusion} \label{sec:conclusion}
We introduced the problem of data composition: organizing a fragmented, heterogeneous data lake into meaningful relations without conditioning on specific targets. We propose an information-driven composition approach that approximates multivariate dependence by identifying sets of attributes that break the independence hypothesis. We map the attributes to a predicate space, which serves as an abstraction layer that decouples multivariate dependency analysis from the data's underlying types. The predicates form a lattice under inclusion, and we contribute a dedicated and scalable algorithm that mines dependent predicate sets. The algorithm is instantiated in \sys, the first system to fully satisfy all four data composition requirements (R1--R4). \sys~is able to efficiently scale beyond classical mining algorithms, thus allowing the detection of dependencies that would otherwise be impractical, and the target-agnostic compositions it produces benefit diverse downstream tasks by providing a subset of low-noise, statistically meaningful attributes. Future work aims to solve the loss of information due to predicate abstraction, add safeguards to prevent composition collapsing into a single component, and extend \sys~to heterogeneous domains and predicate spaces other than tuple pairs.


\begin{acks}
This work has been partly supported by the Horizon Europe Programme under GA.101135513 (CyclOps) and the Spanish Ministerio de Ciencia e Innovación under project PID2023-152841OA-I00 / AEI/10.13039/501100011033 (TALC). Marc Maynou is supported by Google's PhD Fellowship Program. Anna Queralt is a Serra Húnter Fellow. Albert Martin is funded by the predoctoral program AGAUR-FI grants (2025 FI-1 00967) Joan Oró, which is backed by the Department of Research and Universities of the Generalitat of Catalonia, as well as the European Social Plus Fund.
\end{acks}


\bibliographystyle{ACM-Reference-Format}
\bibliography{sample}

@article{nargesian2019data,
  title={Data lake management: challenges and opportunities},
  author={Nargesian, Fatemeh and Zhu, Erkang and Miller, Ren{\'e}e J and Pu, Ken Q and Arocena, Patricia C},
  journal={Proceedings of the VLDB Endowment},
  volume={12},
  number={12},
  pages={1986--1989},
  year={2019},
  publisher={VLDB Endowment}
}

@article{hai2023data,
  title={Data lakes: A survey of functions and systems},
  author={Hai, Rihan and Koutras, Christos and Quix, Christoph and Jarke, Matthias},
  journal={IEEE Transactions on Knowledge and Data Engineering},
  volume={35},
  number={12},
  pages={12571--12590},
  year={2023},
  publisher={IEEE}
}

@online{databricks2023medallion,
  author  = {{Databricks}},
  title   = {What Is the Medallion Lakehouse Architecture?},
  year    = {2023},
  url     = {https://www.databricks.com/glossary/medallion-architecture},
  urldate = {2026-06-21}
}

@article{abadi2020seattle,
  title={The Seattle report on database research},
  author={Abadi, Daniel and Ailamaki, Anastasia and Andersen, David and Bailis, Peter and Balazinska, Magdalena and Bernstein, Philip and Boncz, Peter and Chaudhuri, Surajit and Cheung, Alvin and Doan, AnHai and others},
  journal={ACM Sigmod Record},
  volume={48},
  number={4},
  pages={44--53},
  year={2020},
  publisher={ACM New York, NY, USA}
}

@inproceedings{fernandez2018seeping,
  title={Seeping semantics: Linking datasets using word embeddings for data discovery},
  author={Fernandez, Raul Castro and Mansour, Essam and Qahtan, Abdulhakim A and Elmagarmid, Ahmed and Ilyas, Ihab and Madden, Samuel and Ouzzani, Mourad and Stonebraker, Michael and Tang, Nan},
  booktitle={2018 IEEE 34th International Conference on Data Engineering (ICDE)},
  pages={989--1000},
  year={2018},
  publisher={IEEE}
}

@article{guyon2003introduction,
  title={An introduction to variable and feature selection},
  author={Guyon, Isabelle and Elisseeff, Andr{\'e}},
  journal={Journal of machine learning research},
  volume={3},
  number={Mar},
  pages={1157--1182},
  year={2003}
}

@article{nargesian2018table,
  title={Table union search on open data},
  author={Nargesian, Fatemeh and Zhu, Erkang and Pu, Ken Q and Miller, Ren{\'e}e J},
  journal={Proceedings of the VLDB Endowment},
  volume={11},
  number={7},
  pages={813--825},
  year={2018},
  publisher={VLDB Endowment}
}

@article{cui2026tabular,
  title={Tabular data augmentation for machine learning: Progress and prospects of embracing generative ai},
  author={Cui, Lingxi and Li, Huan and Chen, Ke and Shou, Lidan and Chen, Gang},
  journal={ACM Computing Surveys},
  volume={58},
  number={12},
  pages={1--39},
  year={2026},
  publisher={ACM New York, NY}
}

@article{huhtala1999tane,
  title={TANE: An efficient algorithm for discovering functional and approximate dependencies},
  author={Huhtala, Yka and K{\"a}rkk{\"a}inen, Juha and Porkka, Pasi and Toivonen, Hannu},
  journal={The computer journal},
  volume={42},
  number={2},
  pages={100--111},
  year={1999},
  publisher={OUP}
}

@article{mitra2002unsupervised,
  title={Unsupervised feature selection using feature similarity},
  author={Mitra, Pabitra and Murthy, Chivukula Anjaneya and Pal, Sankar K.},
  journal={IEEE transactions on pattern analysis and machine intelligence},
  volume={24},
  number={3},
  pages={301--312},
  year={2002},
  publisher={IEEE}
}

@article{solorio2020review,
  title={A review of unsupervised feature selection methods},
  author={Solorio-Fern{\'a}ndez, Sa{\'u}l and Carrasco-Ochoa, J Ariel and Mart{\'\i}nez-Trinidad, Jos{\'e} Fco},
  journal={Artificial intelligence review},
  volume={53},
  number={2},
  pages={907--948},
  year={2020},
  publisher={Springer}
}

@article{he2005laplacian,
  title={Laplacian score for feature selection},
  author={He, Xiaofei and Cai, Deng and Niyogi, Partha},
  journal={Advances in neural information processing systems},
  volume={18},
  year={2005}
}

@inproceedings{zhao2007spectral,
  title={Spectral feature selection for supervised and unsupervised learning},
  author={Zhao, Zheng and Liu, Huan},
  booktitle={Proceedings of the 24th international conference on Machine learning},
  pages={1151--1157},
  year={2007}
}

@inproceedings{cai2010unsupervised,
  title={Unsupervised feature selection for multi-cluster data},
  author={Cai, Deng and Zhang, Chiyuan and He, Xiaofei},
  booktitle={Proceedings of the 16th ACM SIGKDD international conference on Knowledge discovery and data mining},
  pages={333--342},
  year={2010}
}

@inproceedings{li2012unsupervised,
  title={Unsupervised feature selection using nonnegative spectral analysis},
  author={Li, Zechao and Yang, Yi and Liu, Jing and Zhou, Xiaofang and Lu, Hanqing},
  booktitle={Proceedings of the AAAI conference on artificial intelligence},
  volume={26},
  number={1},
  pages={1026--1032},
  year={2012}
}

@inproceedings{DBLP:conf/ijcai/YangSMHZ11,
  author    = {Yang, Yi and Shen, Heng Tao and Ma, Zhigang and Huang, Zi and Zhou, Xiaofang},
  title     = {$\ell_{2,1}$-Norm Regularized Discriminative Feature Selection for Unsupervised Learning},
  booktitle = {IJCAI international joint conference on artificial intelligence},
  year      = {2011},
  pages     = {1589--1594}
}

@article{watanabe1960information,
  title={Information theoretical analysis of multivariate correlation},
  author={Watanabe, Satosi},
  journal={IBM Journal of research and development},
  volume={4},
  number={1},
  pages={66--82},
  year={1960},
  publisher={IBM}
}

@article{DBLP:journals/tit/McGill54,
  author       = {William J. McGill},
  title        = {Multivariate information transmission},
  journal      = {Trans. {IRE} Prof. Group Inf. Theory},
  volume       = {4},
  pages        = {93--111},
  year         = {1954},
  url          = {https://doi.org/10.1109/TIT.1954.1057469},
  doi          = {10.1109/TIT.1954.1057469},
  bibsource    = {dblp computer science bibliography, https://dblp.org}
}

@inproceedings{bell2003co,
  title={The co-information lattice},
  author={Bell, Anthony J},
  booktitle={Proceedings of the fifth international workshop on independent component analysis and blind signal separation: ICA},
  volume={2003},
  year={2003}
}

@article{fan2008sure,
  title={Sure independence screening for ultrahigh dimensional feature space},
  author={Fan, Jianqing and Lv, Jinchi},
  journal={Journal of the Royal Statistical Society Series B: Statistical Methodology},
  volume={70},
  number={5},
  pages={849--911},
  year={2008},
  publisher={Oxford University Press}
}

@article{fan2018discoveries,
  title={Are discoveries spurious? Distributions of maximum spurious correlations and their applications},
  author={Fan, Jianqing and Shao, Qi-Man and Zhou, Wen-Xin},
  journal={Annals of statistics},
  volume={46},
  number={3},
  pages={989--1018},
  year={2018}
}

@inproceedings{bogatu2020dataset,
  title={Dataset discovery in data lakes},
  author={Bogatu, Alex and Fernandes, Alvaro AA and Paton, Norman W and Konstantinou, Nikolaos},
  booktitle={2020 ieee 36th international conference on data engineering (icde)},
  pages={709--720},
  year={2020},
  publisher={IEEE}
}

@article{vergara2014review,
  title={A review of feature selection methods based on mutual information},
  author={Vergara, Jorge R and Est{\'e}vez, Pablo A},
  journal={Neural computing and applications},
  volume={24},
  number={1},
  pages={175--186},
  year={2014},
  publisher={Springer}
}

@article{brown2012conditional,
  title={Conditional likelihood maximisation: a unifying framework for information theoretic feature selection},
  author={Brown, Gavin and Pocock, Adam and Zhao, Ming-Jie and Luj{\'a}n, Mikel},
  journal={The journal of machine learning research},
  volume={13},
  pages={27--66},
  year={2012},
  publisher={JMLR. org}
}

@article{huang2021adaptive,
  title={Adaptive graph-based generalized regression model for unsupervised feature selection},
  author={Huang, Yanyong and Shen, Zongxin and Cai, Fuxu and Li, Tianrui and Lv, Fengmao},
  journal={Knowledge-Based Systems},
  volume={227},
  pages={107156},
  year={2021},
  publisher={Elsevier}
}

@inproceedings{roffo2015infinite,
  title={Infinite feature selection},
  author={Roffo, Giorgio and Melzi, Simone and Cristani, Marco},
  booktitle={Proceedings of the IEEE international conference on computer vision},
  pages={4202--4210},
  year={2015}
}

@inproceedings{lu2007feature,
  title={Feature selection using principal feature analysis},
  author={Lu, Yijuan and Cohen, Ira and Zhou, Xiang Sean and Tian, Qi},
  booktitle={Proceedings of the 15th ACM international conference on Multimedia},
  pages={301--304},
  year={2007}
}

@article{zhu2015unsupervised,
  title={Unsupervised feature selection by regularized self-representation},
  author={Zhu, Pengfei and Zuo, Wangmeng and Zhang, Lei and Hu, Qinghua and Shiu, Simon CK},
  journal={Pattern Recognition},
  volume={48},
  number={2},
  pages={438--446},
  year={2015},
  publisher={Elsevier}
}

@inproceedings{balin2019concrete,
  title={Concrete autoencoders: Differentiable feature selection and reconstruction},
  author={Bal{\i}n, Muhammed Fatih and Abid, Abubakar and Zou, James},
  booktitle={International conference on machine learning},
  pages={444--453},
  year={2019},
  organization={PMLR}
}

@article{yuan2021novel,
  title={A novel unsupervised approach to heterogeneous feature selection based on fuzzy mutual information},
  author={Yuan, Zhong and Chen, Hongmei and Zhang, Pengfei and Wan, Jihong and Li, Tianrui},
  journal={IEEE Transactions on fuzzy systems},
  volume={30},
  number={9},
  pages={3395--3409},
  year={2022},
  publisher={IEEE}
}

@article{lim2021pairwise,
  title={Pairwise dependence-based unsupervised feature selection},
  author={Lim, Hyunki and Kim, Dae-Won},
  journal={Pattern Recognition},
  volume={111},
  pages={107663},
  year={2021},
  publisher={Elsevier}
}

@article{DBLP:journals/tkde/MaynouNPFRQ26,
  author       = {Marc Maynou and
                  Sergi Nadal and
                  Raquel Panadero and
                  Javier Flores and
                  Oscar Romero and
                  Anna Queralt},
  title        = {Freyja: Efficient Join Discovery in Data Lakes},
  journal      = {{IEEE} Trans. Knowl. Data Eng.},
  volume       = {38},
  number       = {4},
  pages        = {2277--2288},
  year         = {2026}
}

@article{DBLP:journals/pvldb/ChepurkoMZFKK20,
  author       = {Nadiia Chepurko and
                  Ryan Marcus and
                  Emanuel Zgraggen and
                  Raul Castro Fernandez and
                  Tim Kraska and
                  David R. Karger},
  title        = {{ARDA:} Automatic Relational Data Augmentation for Machine Learning},
  journal      = {Proc. {VLDB} Endow.},
  volume       = {13},
  number       = {9},
  pages        = {1373--1387},
  year         = {2020}
}

@article{DBLP:journals/tmlr/CappuzzoCLPV25,
  author       = {Riccardo Cappuzzo and
                  Aimee Coelho and
                  F{\'{e}}lix Lefebvre and
                  Paolo Papotti and
                  Ga{\"{e}}l Varoquaux},
  title        = {Retrieve, Merge, Predict: Augmenting Tables with Data Lakes},
  journal      = {Trans. Mach. Learn. Res.},
  volume       = {2025},
  year         = {2025}
}

@inproceedings{DBLP:conf/icde/IonescuVB0K24,
  author       = {Andra Ionescu and
                  Kiril Vasilev and
                  Florena Buse and
                  Rihan Hai and
                  Asterios Katsifodimos},
  title        = {AutoFeat: Transitive Feature Discovery over Join Paths},
  booktitle    = {{ICDE}},
  pages        = {1861--1873},
  publisher    = {{IEEE}},
  year         = {2024}
}

@inproceedings{DBLP:conf/sigmod/IlyasMHBA04,
  author       = {Ihab F. Ilyas and
                  Volker Markl and
                  Peter J. Haas and
                  Paul Brown and
                  Ashraf Aboulnaga},
  title        = {{CORDS:} Automatic Discovery of Correlations and Soft Functional Dependencies},
  booktitle    = {{SIGMOD} Conference},
  pages        = {647--658},
  publisher    = {{ACM}},
  year         = {2004}
}

@inproceedings{DBLP:conf/sigmod/ZhangGR20,
  author       = {Yunjia Zhang and
                  Zhihan Guo and
                  Theodoros Rekatsinas},
  title        = {A Statistical Perspective on Discovering Functional Dependencies in
                  Noisy Data},
  booktitle    = {{SIGMOD} Conference},
  pages        = {861--876},
  publisher    = {{ACM}},
  year         = {2020}
}

@inproceedings{DBLP:conf/sigmod/KenigMPSS20,
  author       = {Batya Kenig and
                  Pranay Mundra and
                  Guna Prasaad and
                  Babak Salimi and
                  Dan Suciu},
  title        = {Mining Approximate Acyclic Schemes from Relations},
  booktitle    = {{SIGMOD} Conference},
  pages        = {297--312},
  publisher    = {{ACM}},
  year         = {2020}
}

@article{DBLP:journals/pvldb/PapenbrockEMNRZ15,
  author       = {Thorsten Papenbrock and
                  Jens Ehrlich and
                  Jannik Marten and
                  Tommy Neubert and
                  Jan{-}Peer Rudolph and
                  Martin Sch{\"{o}}nberg and
                  Jakob Zwiener and
                  Felix Naumann},
  title        = {Functional Dependency Discovery: An Experimental Evaluation of Seven
                  Algorithms},
  journal      = {Proc. {VLDB} Endow.},
  volume       = {8},
  number       = {10},
  pages        = {1082--1093},
  year         = {2015}
}

@inproceedings{DBLP:conf/icde/FernandezAKYMS18,
  author       = {Raul Castro Fernandez and
                  Ziawasch Abedjan and
                  Famien Koko and
                  Gina Yuan and
                  Samuel Madden and
                  Michael Stonebraker},
  title        = {Aurum: {A} Data Discovery System},
  booktitle    = {{ICDE}},
  pages        = {1001--1012},
  publisher    = {{IEEE} Computer Society},
  year         = {2018}
}

@article{DBLP:journals/tkde/NargesianPBZM23,
  author       = {Fatemeh Nargesian and
                  Ken Q. Pu and
                  Bahar Ghadiri Bashardoost and
                  Erkang Zhu and
                  Ren{\'{e}}e J. Miller},
  title        = {Data Lake Organization},
  journal      = {{IEEE} Trans. Knowl. Data Eng.},
  volume       = {35},
  number       = {1},
  pages        = {237--250},
  year         = {2023}
}

@inproceedings{DBLP:conf/sigmod/AgrawalIS93,
  author       = {Rakesh Agrawal and
                  Tomasz Imielinski and
                  Arun N. Swami},
  title        = {Mining Association Rules between Sets of Items in Large Databases},
  booktitle    = {{SIGMOD} Conference},
  pages        = {207--216},
  publisher    = {{ACM} Press},
  year         = {1993}
}

@inproceedings{DBLP:conf/sigmod/HanPY00,
  author       = {Jiawei Han and
                  Jian Pei and
                  Yiwen Yin},
  title        = {Mining Frequent Patterns without Candidate Generation},
  booktitle    = {{SIGMOD} Conference},
  pages        = {1--12},
  publisher    = {{ACM}},
  year         = {2000}
}

@inproceedings{DBLP:conf/vldb/Toivonen96,
  author       = {Hannu Toivonen},
  title        = {Sampling Large Databases for Association Rules},
  booktitle    = {{VLDB}},
  pages        = {134--145},
  publisher    = {Morgan Kaufmann},
  year         = {1996}
}

@inproceedings{DBLP:conf/cikm/RiondatoDFU12,
  author       = {Matteo Riondato and
                  Justin A. DeBrabant and
                  Rodrigo Fonseca and
                  Eli Upfal},
  title        = {{PARMA:} a parallel randomized algorithm for approximate association
                  rules mining in MapReduce},
  booktitle    = {{CIKM}},
  pages        = {85--94},
  publisher    = {{ACM}},
  year         = {2012}
}

@article{DBLP:journals/ml/Webb07,
  author       = {Geoffrey I. Webb},
  title        = {Discovering Significant Patterns},
  journal      = {Mach. Learn.},
  volume       = {68},
  number       = {1},
  pages        = {1--33},
  year         = {2007}
}

@inproceedings{DBLP:conf/kdd/PellegrinaV18,
  author       = {Leonardo Pellegrina and
                  Fabio Vandin},
  title        = {Efficient Mining of the Most Significant Patterns with Permutation
                  Testing},
  booktitle    = {{KDD}},
  pages        = {2070--2079},
  publisher    = {{ACM}},
  year         = {2018}
}

@article{martin_discovering_2025,
	title = {Discovering {Approximate} {Denial} {Constraints} in {Large} {Databases}},
	volume = {19},
	issn = {2150-8097},
	number = {9},
	urldate = {2026-08-10},
	journal = {Proceedings of the VLDB Endowment},
	author = {Martin, Albert and De Almeida, Eduardo C. and Romero, Oscar and Queralt, Anna},
	month = {july},
	year = {2026},
	pages = {2086 -- 2098},
}

@article{martin_how_2025,
	title = {How and {Why} {False} {Denial} {Constraints} are {Discovered}},
	volume = {18},
	issn = {2150-8097},
	language = {en},
	number = {10},
	urldate = {2025-09-10},
	journal = {Proceedings of the VLDB Endowment},
	author = {Martin, Albert and De Almeida, Eduardo C. and Romero, Oscar and Queralt, Anna},
	month = {jun},
	year = {2025},
	pages = {3477--3489},
}

@article{AR_sample_riondato2014efficient,
  title={Efficient discovery of association rules and frequent itemsets through sampling with tight performance guarantees},
  author={Riondato, Matteo and Upfal, Eli},
  journal={ACM Transactions on Knowledge Discovery from Data (TKDD)},
  volume={8},
  number={4},
  pages={1--32},
  year={2014},
  publisher={ACM New York, NY, USA}
}

@book{bishop2006pattern,
  title={Pattern recognition and machine learning},
  author={Bishop, Christopher M},
  year={2006},
  publisher={Springer}
}

@article{chu2013discovering,
  title={Discovering denial constraints},
  author={Chu, Xu and Ilyas, Ihab F and Papotti, Paolo},
  journal={Proceedings of the VLDB Endowment},
  volume={6},
  number={13},
  pages={1498--1509},
  year={2013},
  publisher={VLDB Endowment}
}

@article{pena2019discovery,
  title={Discovery of approximate (and exact) denial constraints},
  author={Pena, Eduardo HM and De Almeida, Eduardo C and Naumann, Felix},
  journal={Proceedings of the VLDB Endowment},
  volume={13},
  number={3},
  pages={266--278},
  year={2019},
  publisher={VLDB Endowment}
}

@article{dy2004feature,
  title={Feature selection for unsupervised learning},
  author={Dy, Jennifer G and Brodley, Carla E},
  journal={Journal of machine learning research},
  volume={5},
  number={Aug},
  pages={845--889},
  year={2004}
}

@inproceedings{armbrust2021lakehouse,
  title={Lakehouse: a new generation of open platforms that unify data warehousing and advanced analytics},
  author={Armbrust, Michael and Ghodsi, Ali and Xin, Reynold and Zaharia, Matei and others},
  booktitle={11th Conference on Innovative Data Systems Research, {CIDR} 2021},
  year={2021}
}

\end{document}